\documentclass[preprint]{elsarticle}
\usepackage[T1]{fontenc}
\usepackage[utf8]{inputenc}
\usepackage[letterpaper, margin=1in]{geometry}
\usepackage[colorlinks]{hyperref}

\usepackage{subfigure}
\usepackage{array}
\usepackage{amsmath,amssymb,amsfonts,xcolor,amsthm,dsfont,,mathabx}
\usepackage{placeins}
\usepackage{csquotes}
\usepackage{latexsym}
\usepackage{euscript}
\usepackage{tocvsec2}
\usepackage{etoolbox,caption}
\usepackage{float}
\usepackage{marginnote}
\usepackage{epsfig,tikz}
\usetikzlibrary{patterns}
\usetikzlibrary{arrows.meta}
\usetikzlibrary{matrix}

\usepackage{enumitem}
\makeatletter
\newcommand{\mylabel}[2]{#2\def\@currentlabel{#2}\label{#1}}
\makeatother

\newtheorem*{theorem-no}{Theorem}

\usepackage{sectsty}
\allsectionsfont{\sffamily}
\usepackage{fancyhdr}
\newcommand{\comment}[1]{}  
\newcounter{pulse}[section]
\numberwithin{pulse}{section}  

\newcommand{\thf}{\sc} 

\theoremstyle{plain}
\newtheorem{theorem}[pulse]{\thf Theorem}

\newtheorem{proposition}[pulse]{\thf Proposition}

\newtheorem{corollary}[pulse]{\thf Corollary}
\theoremstyle{definition}
\newtheorem{definition}[pulse]{\thf Definition}
\newtheorem*{notation}{\thf Notation}

\newtheorem{example}[pulse]{\thf Example}
\theoremstyle{remark}
\newtheorem{remark}[pulse]{\thf Remark}

\def\t{\rm t}

\def\id{\rm id}

\def\t{\rm t}

\def\bbC{\mathbb C}
\def\bbG{\mathbb G}
\def\bbZ{\mathbb Z}
\def\bb1{\mathds 1}

\newcommand{\bbS}{\mathbb{S}}

\newcommand{\C}{\mathbb{C}}

\newcommand{\la}{\langle}
\newcommand{\ra}{\rangle}

\def\U{\mathcal U}

\def\E{\mathcal E}

\newcommand{\red}[1]{{\color{black} #1}}

\renewcommand\paragraph[1]{\par\vspace{1em}\noindent\textbf{#1}}

\begin{document}
\begin{frontmatter}
\title{Frames for signal processing on Cayley graphs}
\author[mymainaddress]{Kathryn Beck}
\author[mymainaddress]{Mahya Ghandehari\corref{mycorrespondingauthor}}
\cortext[mycorrespondingauthor]{Corresponding author}
\ead{mahya@udel.edu}
\author[mysecondaryaddress]{Skyler Hudson}
\author[mymainaddress]{Jenna Paltenstein}
\address[mymainaddress]{Department of Mathematical Sciences, University of Delaware, Newark, DE, USA, 19716}
\address[mysecondaryaddress]{Department of Mathematics \& Statistics, University of Massachusetts at Amherst, 
Amherst, MA, USA, 01003}
%
%
\begin{abstract}
The spectral decomposition of graph adjacency matrices is an essential ingredient  in the design of graph signal processing (GSP) techniques.
When the adjacency matrix has multi-dimensional eigenspaces, it is desirable to base GSP constructions on a particular eigenbasis \red{that better reflects the graph's symmetries}.
In this paper, we provide an explicit and detailed representation-theoretic account for the spectral decomposition of the adjacency matrix of a weighted Cayley graph.
Our method applies to all weighted Cayley graphs, regardless of whether they are quasi-Abelian, and offers detailed descriptions of eigenvalues and eigenvectors derived from the coefficient functions of the representations of the underlying group. 
\red{Next, we turn our attention to constructing frames on Cayley graphs.
Frames are overcomplete spanning sets that ensure stable and potentially redundant systems for signal reconstruction.
We use our proposed eigenbases to build frames that are suitable for developing signal processing on Cayley graphs.} 
These are the Frobenius--Schur frames and Cayley frames, for which we provide a characterization and a practical recipe for their construction. 
\end{abstract}
\end{frontmatter}

\section{Introduction}\label{section:intro}
Graph signal processing (GSP) is a fast-growing field that offers a framework for developing signal processing techniques tailored for signals that are defined on graphs, with the objective of incorporating the underlying graph structure into the analysis.
For a fixed graph $G$ with vertex set $V$, a graph signal on $G$ is a complex-valued function $f:V\to \bbC$. If the vertex set $V$ is labeled, say $\{v_i\}_{i=1}^N$, then the graph signal can be represented as a column vector $\left[f(v_1), f(v_2),\ldots, f(v_N)\right]^{\t}$ in ${\mathbb C}^N$, where $\t$ denotes the matrix transpose.  
A powerful technique to analyze graph signals that has gained significant popularity over the recent years involves fixing a basis of eigenvectors for an appropriate matrix associated with the graph; we call such a basis a \emph{graph Fourier basis}. Expanding graph signals in this basis leads to the idea of Fourier analysis on graphs. The reason for this approach is to improve signal processing efficiency by working with a basis that encodes the structural features of the underlying graph, rather than an arbitrary basis of ${\mathbb C}^N$. 

Prominent examples of matrices associated with graphs are the graph adjacency matrix and the graph Laplacian.
\red{In this manuscript, we focus on weighted graphs, where each edge is assigned a numerical value known as a \emph{weight}, while non-edges are assigned a weight of zero.
The adjacency matrix of a weighted graph $G$ with $N$ nodes is the matrix $A_G$ of size $N$, whose $(i,j)$-th entry is precisely the weight of the edge joining nodes $i$ and $j$.
The Laplacian of $G$, denoted by $L_G$, is an $N\times N$ matrix defined as $L_G=D_G-A_G$, where $D_G$ is the diagonal matrix with entries $d_{ii}$ representing the degree of the vertex $v_i$ (the sum of the weights of edges incident to $v_i$).
We note that unweighted graphs can be viewed as a specific type of weighted graphs in which each edge is assigned a weight of 1.}
%

\red{Consider either the graph adjacency matrix $A_G$ or the graph Laplacian $L_G$; these matrices are often referred to as the \emph{graph shift operator}}. Next, fix an orthonormal basis of eigenvectors $\phi_1,\ldots,\phi_N\in \bbC^N$ associated with (possibly repeated) eigenvalues $\lambda_1,\ldots,\lambda_N$ for that matrix. 
The {\it graph Fourier transform} $\widehat{f}$ of a graph signal $f:V\to \bbC$ is defined to be the expansion of $f$ in terms of this orthonormal basis. Namely, 
\begin{equation*}\label{GFT}
\widehat{f}(\phi_i) = \langle f, \phi_i\rangle = \sum_{n=1}^N f(v_n)\overline{\phi_i(v_n)}.
\end{equation*}
The corresponding {\it inverse Fourier transform} is given by
\begin{equation*}\label{GIFT}
f(v_n) = \sum_{i=1}^{N} \widehat f (\phi_i) \phi_i(v_n).
\end{equation*}
Here, $\langle\cdot,\cdot\rangle$ denotes the inner product on $\bbC^N$. 
It is worth mentioning that the eigenvectors $\{\phi_i\}_{i=1}^N$ can be chosen from ${\mathbb R}^N$, 
since both adjacency and Laplacian matrices are symmetric and real-valued, and therefore, they both have real spectrum, and are diagonalizable over ${\mathbb R}$. However, we do not impose such a restriction on our choice of eigenbasis. In fact, we will show that in certain cases, a complex eigenbasis may provide us with a more efficient Fourier analysis (see e.g.~Proposition~\ref{prop:eigenvector} and Corollary~\ref{cor:constant-on-conjugacy}).
We refer to~\cite{2013:Sandryhaila:DS,2014:Sandryhaila:BD,SNFOV} for a detailed background on the graph Fourier transform, to ~\cite{2018:Ortega:GSPOverview,ortega2022} for a general overview of graph signal processing, and to \cite{Ruizetal,GJK-2022} for some new developments in signal processing on large graphs and graphons.

Taking the graph Fourier transform as defined above as a first step, a significant body of research has been devoted to generalizing classical tools from Fourier analysis to the case of signals defined on graphs. Important examples of such efforts include wavelet constructions (e.g.~\cite{Hammond:2011:WaveletsOnGraphViaSpectralGraphTheory,2009:Jansen:MultiscaleMethodGraphData,2014:Chui:RepsFunctionsBigData}), frame constructions (e.g.~\cite{2013:Shuman:SpectrumAdaptedWavelet,Shuman:2016:VertexFrequencyAnalysisOnGraphs}), constructions of wavelet-type frames (e.g.~\cite{2017:Dong:SparseRepWavelet,2018:Gobel:TightFramesDenoising,2013:Leonardi:TightWaveletFrames,2013:Shuman:SpectrumAdaptedWavelet}),
and constructions of Gabor-type frames (e.g.~\cite{2016:Behjat:SignalAdaptedFrames,  2021:Ghandehari:GaborTypeFrames})
for graph signals.  
Normally, wavelet and frame constructions, as well as many other signal processing concepts, rely heavily on the choice of the Fourier basis. Thus,  a suitable selection of eigenbasis for the graph Fourier transform plays a critical role in the success of this theory. 
The significance of this phenomenon is accentuated when dealing with a graph (adjacency or Laplacian) matrix with high-dimensional eigenspaces. A prominent example of such a scenario is the case of a Cayley graph, particularly one that is associated with a non-Abelian group.

\red{A weighted Cayley graph has vertices corresponding to elements of a group $\bbG$ and weighted edges generated by an inverse-invariant function $w:\bbG\to [0,\infty)$, called a weight function (see Definition~\ref{def-cayley}).  When $w$ is $\{0,1\}$-valued, this definition reduces to the definition of an unweighted Cayley graph.}
The underlying algebraic structure and highly symmetric nature of (weighted) Cayley graphs make them a rich category of graphs for various applications, leading to the need for the advancement of further graph signal processing techniques for this class. 
For examples of signal processing on Cayley graphs, see \cite{Rockmore, paper} for the case of Cayley graphs of the symmetric group, and see \cite{GSP-circulant2019,Kotzagiannidis2018,Chepuri2017} for the case of circulant graphs, which are Cayley graphs on $\bbZ_n$. 

For Cayley graphs (or any regular graph in general), the eigenbasis of the associated adjacency and Laplacian matrices are identical. 
So, in this article, we focus our attention only to adjacency matrices. 
Eigen-decompositions of adjacency matrices of Cayley graphs are well-understood when the group is Abelian (\cite{1979:Babai:Spectraofcayleygraphs}) or \red{the Cayley graph is} quasi-Abelian,\footnote{\red{The terminology ``quasi-Abelian'' is somewhat misleading, as it does not pertain to the underlying group of a weighted Cayley graph; instead, it relates to the conjugation-invariance of its weight function $w$ (see the definition before Corollary~\ref{cor:constant-on-conjugacy}). }}
meaning the generating set is closed under conjugation (\cite{Rockmore}). 
In \cite{MDK:2019:Sampta}, the second author and collaborators use representation theory of groups to construct a suitable Fourier basis (i.e.,~eigenbasis of the graph adjacency matrix) for signal processing on quasi-Abelian Cayley graphs.
They contend that the particular eigenbasis constructed through representation theoretic considerations is more suitable for developing the Fourier transform on a Cayley graph. For example, using these eigenbases simplifies several operations on graph signals including the graph translation operator.
Moreover, they show that such eigenbases can be used to construct a suitable family of tight frames.

In \cite{paper}, Chen et al.~take a similar viewpoint when studying ranked data sets as signals on the permutahedron.
The permutahedron, denoted by $\mathbb P_n$, is the Cayley graph of the symmetric group with the generating set of adjacent transpositions. 
Selecting the permutahedron as the underlying graph is crucial for the success of their theory, as the generating set of the permutahedron captures a specific notion of distance that is useful in the context of ranked voting systems.
%
Chen et al.~construct a special basis of eigenvectors for the vector space $\ell^2(\bbS_n)$ that is compatible with both irreducible representations of $\bbS_n$ and eigenspaces of the adjacency matrix of ${\mathbb P}_n$.
They use this basis to form a frame (i.e., an overcomplete spanning set) for $\ell^2(\bbS_n)$ that is suitable for GSP on the permutahedron.
The significance of obtaining such a frame is that the analysis coefficients, which are the inner products of vectors in $\ell^2(\bbS_n)$ with frame elements, provide specific interpretations of the ranked data
(e.g.~popularity of candidates, when a candidate is polarizing, and when two candidates are likely to be ranked similarly).
This example highlights the importance of constructing appropriate Fourier bases for (not necessarily quasi-Abelian) Cayley graphs.

In \cite{EUSIPCO}, the first two authors generalize the results of Chen et al.~to all Cayley graphs on $\bbS_n$. 
Namely, they introduce a class of frames, called Frobenius--Schur frames, which have the property that every frame vector belongs to the coefficient space of only one irreducible representation of the symmetric group. 
Furthermore, they characterize all Frobenius--Schur frames on the group algebra of the symmetric group which are `compatible' with respect to both the generating set and the representation theory of the group. 
They observe that frames obtained in \cite{paper} are exactly such compatible Frobenius--Schur frames (which we call Cayley frames in this paper); see Subsection~\ref{subsec:relation-to-paper} for a detailed explanation.  

\subsection{Main contribution}
In the present article, we take the perspective in \cite{MDK:2019:Sampta} to extend the results of \cite{EUSIPCO} to general \red{weighted} Cayley graphs; that is weighted Cayley graphs on any group $\bbG$, with any weight function. This includes the case of unweighted Cayley graphs on any group with any inverse-closed generating set.
The definitions of Frobenius--Schur frames and frames compatible with a generating set can be naturally extended in the case of general \red{weighted} Cayley graphs. Frobenius--Schur frames are those that are compatible with the representation theory of the underlying group $\bbG$; whereas, we introduce weighted Cayley frames with respect to a given weight function $w$ as those that are compatible with $w$ and the representation theory of $\bbG$ (Definition~\ref{def:frame-compatible}). 

\red{Our contribution in this paper is two-fold. Firstly, we provide a complete description of the eigen-decomposition of the adjacency matrix of a weighted Cayley graph in terms of the irreducible representations of its underlying group (Proposition~\ref{prop:eigenvector}). This work generalizes existing results on the spectral decomposition of Cayley graphs on Abelian groups \cite{1979:Babai:Spectraofcayleygraphs} and quasi-Abelian Cayley graphs \cite{Rockmore}.
Secondly, we characterize} all (weighted) Cayley frames (with respect to a given weight function) of $\ell^2(\bbG)$, and provide a concrete recipe for constructing such frames (Theorem~\ref{thm: frame}).
Given their compatibility with both the group and the weight function, these are suitable frames for signal processing on weighted Cayley graphs.
\red{Our frame construction in Theorem~\ref{thm: frame} is based on the particular eigenbasis provided in Proposition~\ref{prop:eigenvector}, where the representation theory of the underlying group guides the choice of basis elements. Namely, in a weighted Cayley frame, each vector belongs to the coefficient space of only one irreducible representation of the underlying group. This is akin to the selection of the Fourier basis in classical signal processing; indeed, the classical Fourier basis contains precisely one coefficient function for each irreducible representation of ${\mathbb Z}_n$.}
%


\red{Our representation theoretic viewpoint, particularly through the application of Frobenius--Schur theory, offers a significant benefit in our approach, as it provides a block diagonalization of the adjacency matrix resulting in much smaller block sizes than the original matrix. For instance, in Example~\ref{example:P4} we discuss a Cayley graph on the symmetric group $\bbS_4$. The associated adjacency matrix in this case is of size 24, but our method only requires the eigen-decomposition of matrices of size at most 3. Additionally, the block-diagonalizing unitary matrix from Frobenius--Schur theory depends solely on the underlying group of a weighted Cayley graph}. This feature greatly enhances the computational efficiency of our proposed method, as \red{the same block-diagonalizing unitary matrix} works for any \red{weighted} Cayley graph over a given group.

\subsection{Organization of the paper}
The rest of this article is organized as follows. In Section~\ref{subsec:notations-rep}, we present the necessary background for the representation theory of finite groups, along with the Schur orthogonality relations and Frobenius--Schur decomposition theorem. In Section~\ref{subsec:eigen-decomp of A}, we use the Frobenius--Schur theorem to provide a complete description of the eigen-decomposition of the adjacency matrix of a weighted Cayley graph in terms of the irreducible representations of its underlying group (Proposition~\ref{prop:eigenvector}). 
As a corollary to Proposition~\ref{prop:eigenvector}, we  provide a new proof for the quasi-Abelian case. We also provide  examples of applications of Proposition~\ref{prop:eigenvector}  for the case of the permutahedrons, $\mathbb P_3$ and $\mathbb P_4$. In Section~\ref{subsec:FS-frames-Cayley-frames}, we use the results of Section~\ref{subsec:eigen-decomp of A} to introduce and characterize weighted Cayley frames, i.e., Frobenius--Schur frames compatible with the weight function of a Cayley graph. We also provide an example of a Cayley frame for $\ell^2(\bbS_4)$. In Section~\ref{subsec:properties} we discuss properties inherited by the larger Cayley frame from the smaller frames, including tight/parseval frames and the restricted isometry property. Next, we show how our frame construction for general Cayley graphs relates to the frames for the permutahedron built in \cite{paper} (Section~\ref{subsec:relation-to-paper}). We end the paper with \ref{Appendix:proofs} which provides a proof of Proposition~\ref{prop:eigenvector}, and \ref{Appendix:rep-Sn} which gives an overview of the theory behind finding the irreducible representations of $\mathbb S_n$. We have included \ref{Appendix:rep-S3} and \ref{Appendix:rep-S4} with matrix representations and coefficient functions of $\bbS_3$ and $\bbS_4$. While this information is known, we have included it here for the ease of the reader, as finding references presenting this material in a suitable form for us turned out to be challenging. 


\section{Notations and background}\label{section:notations}

\subsection{Representation theory of finite groups}\label{subsec:notations-rep}
In this subsection, we provide the necessary background for the representation theory of finite groups and their associated function spaces. 
We restrict our attention to unitary representations. This is non-consequential as every representation of a finite group is unitizable by a change of inner product on the representation space (see for example  \cite[Section 1.3]{Serre}).
Throughout this article, let $\bbG$ be a finite (not necessarily Abelian)  group of size $N$. A unitary representation of $\bbG$ of dimension $d$ is a group homomorphism $\pi: \bbG\rightarrow \U_d(\bbC)$, where $\U_d(\bbC)$ denotes the (multiplicative) group of unitary matrices of size $d$. \red{For a representation $\pi: \bbG\rightarrow \U_d(\bbC)$, a subspace $W$ of $\C^d$, and an element $g\in \bbG$, define
$$\pi(g)W:=\{\pi(g)\xi: \ \xi\in W\}.$$ 
The subspace $W$ is called $\pi$-invariant if for all $g\in \bbG$, we have $\pi(g)W\subseteq W$.}
Given a $\pi$-invariant subspace $W$, we define the subrepresentation $\pi|_{W}$ of $\pi$ to be the representation of ${\mathbb G}$ on the inner product space $W$ obtained by restricting both the domain and codomain of $\pi(\cdot)$ to $W$.
A representation $\pi$ is called irreducible if $\{0\}$ and $\bbC^d$ are its only $\pi$-invariant subspaces. \red{We say two representations $\pi$ and $\sigma$ of $\bbG$ are unitarily equivalent if there exists a unitary matrix $U$ such that $U^{-1}\pi(g)U=\sigma(g)$ for all $g\in \bbG$.} We let $\widehat{\bbG}$ denote the collection of all (equivalence classes of) irreducible unitary representations of $\bbG$. In the case of an Abelian group, every irreducible representation of $\bbG$ is one-dimensional \cite[Corollary 3.6]{1995:Folland:HarmonicAnalysis},  and $\widehat{\bbG}$ reduces to the group of characters on $\bbG$. 
Note that every unitary representation of a finite group decomposes into a direct sum of irreducible representations in $\widehat{\bbG}$ \cite[Theorem 5.2]{1995:Folland:HarmonicAnalysis}.

An important representation of a group $\bbG$ is its \emph{right regular representation}
$\rho:\bbG\rightarrow \U_N(\bbC)$, where $\rho(g)$ denotes the matrix associated with the permutation $h\mapsto hg^{-1}$, for every $h\in \bbG$. 
The representation $\rho$ \red{is unitarily equivalent to the representation $\rho'$ on the vector space $\ell^2(\bbG)$, defined as follows:}
\begin{equation*}
    \rho':\bbG\rightarrow \U(\ell^2(\bbG)), \quad\rho'(x)f(y)=f(yx), \ \forall f\in \ell^2(\bbG), \ \forall x,y\in \bbG.
\end{equation*}
Here, $\ell^2(\bbG)$ is equipped with the inner product $\langle f,g\rangle_{\ell^2(\bbG)}=\sum_{x\in\bbG} f(x)\overline{g(x)}$, and 
$\U(\ell^2(\bbG))$ denotes the set of unitary operators on the inner product space $\ell^2(\bbG)$. For the rest of this article, we assume that the group elements are labeled as $\bbG=\{g_1,\ldots, g_N\}$. So, the unitary operator $U:\ell^2(\bbG)\to \bbC^N$ defined as $U(f)=\Big[f(g_1),\ldots, f(g_N)\Big]^{\t}$, where $\t$ denotes the matrix transpose, provides the unitary equivalence between $\rho$ and $\rho'$.

When $N>1$, the representation $\rho$ is not irreducible, and thus can be decomposed into a direct sum of irreducible representations. 
In many applications, such as those studied in this paper, it is important to understand the above direct sum decomposition of $\rho$ in concrete terms. 
For general (not necessarily finite) compact groups, the Peter--Weyl theorem provides us with a convenient way to tackle this task. 
This theorem, which was first proved by  Frobenius and Schur for the case of finite groups, gives rise to a decomposition of $\bbC^N$ (respectively $\ell^2(\bbG)$) into spaces generated by coefficient functions.

For an arbitrary $\pi\in\widehat{\bbG}$ of dimension $d_\pi$, and vectors $\xi,\eta\in\bbC^{d_\pi}$, we define the \emph{coefficient function} associated with the representation $\pi$ and the vectors $\xi, \eta$ as follows:
$$\pi_{\xi,\eta}:\bbG\to \bbC, \ \ \pi_{\xi,\eta}(g)=\langle\pi(g)\xi, \eta\rangle, \ \forall g\in \bbG.$$
For the rest of this article, we view $\pi_{\xi,\eta}$ as a function on $\bbG$ or as a vector in $\bbC^N$ interchangeably. Namely, we sometimes think of $\pi_{\xi,\eta}$ as 
$$\pi_{\xi,\eta}=\Big[\langle\pi(g_1)\xi, \eta\rangle, \ldots, \langle\pi(g_N)\xi, \eta\rangle\Big]^{\t}.$$
When $\{e_i\}_{i=1}^{d_\pi}$ is the standard orthonormal basis for $\bbC^{d_\pi}$, the coefficient functions 
$$\pi_{i,j}(x):=\pi_{e_i, e_j}(x)=\la \pi(x)e_i,e_j \ra, \ i,j=1,\ldots, d_\pi$$
indicate the entries of the matrix of $\pi(x)$ represented in the same basis. 
Coefficient functions play a central role in the harmonic analysis of non-Abelian groups. 
For $\pi\in\widehat{\bbG}$ and $1\leq i\leq d_\pi$, define
\begin{equation}
\label{eq:def-E-pi,i}
 {\cal E}_{\pi,i}=\big\{\pi_{\xi,e_i}:\xi\in \bbC^{d_\pi}\big\}.
\end{equation}
\red{For $\pi_{\xi,e_i}\in {\cal E}_{\pi,i}$ and $g,x\in \bbG$, we have
$$\rho(g)\pi_{\xi,e_i}(x)=\pi_{\xi,e_i}(xg)=\langle\pi(xg)\xi,e_i\rangle=\langle\pi(x)\pi(g)\xi,e_i\rangle=\pi_{\pi(g)\xi,e_i}(x).$$
So for every $g\in\bbG$, we get $\rho(g){\cal E}_{\pi,i}\subseteq {\cal E}_{\pi,i}$. Consequently,} every set  ${\cal E}_{\pi,i}$ forms a $\rho$-invariant subspace of $\bbC^N$. 
This fact is used in the well-known Frobenius--Schur theorem,  stating the decomposition of $\rho$. 
Important pieces for the Frobenius--Schur theorem are the following orthogonality relations (\cite[Theorem 3.34]{1995:Folland:HarmonicAnalysis}).
\begin{theorem}[The Schur orthogonality relations]\label{thm:schur-ortho}
Let $\pi,\sigma$ be irreducible unitary representations of $\bbG$, and consider ${\cal E}_{\pi,i}$ and ${\cal E}_{\sigma, j}$ as subspaces of $\bbC^N$. 
\begin{enumerate}
    \item If $\pi$ and $\sigma$ are not unitarily equivalent then ${\cal E}_{\pi,i} \perp {\cal E}_{\sigma,j}$ for all $1\leq i\leq d_\pi$ and $1\leq j\leq d_\sigma$.
    \item If $\{e_j\}_{j=1}^{d_\pi}$ is an orthonormal basis for $\bbC^{d_\pi}$ then $\left\{\sqrt{\frac{d_\pi}{|\bbG|}} \pi_{j,i}:j=1,..,d_\pi \right\}$ is an orthonormal basis for ${\cal E}_{\pi,i}$.
\end{enumerate}
\end{theorem}
This is to say that if $\pi$ and $\sigma$ are \red{not unitarily equivalent} then $\la \pi_{i,j},\sigma_{r,s}\ra_{\bbC^N}=0$ for all $1 \leq i, j \leq d_{\pi}$ and $1 \leq r,s \leq d_{\sigma}$ while 
\[\la \pi_{i,j} , \pi_{r,s} \ra_{\bbC^N} = \frac{|\bbG|}{d_{\pi}} \delta_{i,r} \delta_{j,s},  \ \mbox{ for all }  1 \leq i, j,r,s \leq d_{\pi},\]
where $\delta$ is the Kronecker delta function. 

We can now state the Frobenius--Schur theorem for finite groups.
\begin{theorem}[Frobenius--Schur theorem] \label{thm:frob-schur}
Let $\bbG$ be a finite group of size $N$. For ${\cal E}_{\pi,i}$ as given in Definition~\eqref{eq:def-E-pi,i}, we have 
$$\bbC^N=\oplus_{\pi \in \widehat{\bbG}}\oplus_{1\leq i\leq d_\pi} {\cal E}_{\pi,i},$$ where each ${\cal E}_{\pi,i}$ has the orthonormal basis
\[ \left\{\phi^{\pi}_{j,i}:=\sqrt{\frac{d_{\pi}}{N}} \pi_{j,i}: j=1,...,d_\pi\right\}.  \]
Moreover, for every $1\leq i\leq d_\pi$, 
the subrepresentation $\rho|_{{\cal E}_{\pi,i}}$ is \red{unitarily equivalent} to $\pi$. Consequently, $\rho$ \red{is unitarily equivalent to} $\oplus_{\pi\in\widehat{\bbG}}\ d_p\cdot\pi$; that is, 
each $\pi\in \widehat{\bbG}$ occurs in the right regular representation of $\bbG$ with multiplicity $d_{\pi}$.
\end{theorem}
We can reformulate Theorem~\ref{thm:frob-schur} as simultaneous block diagonalization of \red{matrices $\rho(g)$ for all $g\in\bbG$}. Consider a fixed ordering of $\widehat{\bbG}$, e.g.~ $\widehat{\bbG}=\{\pi^1,\pi^2,\ldots,\pi^m\}$. Let $d_i$ denote the dimension of the representation $\pi^i$.
Let $B$ be a matrix of size $N=|\bbG|$ whose columns are the vectors of the orthonormal basis in Theorem \ref{thm:frob-schur} ordered appropriately. Namely, 
\begin{equation}\label{eq:matrix-B}
 B=\Big[
 \underbrace{\phi^{1}_{1,1}|\ldots |\phi^{1}_{d_1,1}}_{\mbox{basis for }{\cal E}_{\pi^1,1}}|
 \underbrace{\phi^{1}_{1,2}|\ldots |\phi^{1}_{d_1,2}}_{\mbox{basis for }{\cal E}_{\pi^1,2}}|
 \ldots
 | \underbrace{\phi^{1}_{1,d_1}|\ldots |\phi^{1}_{d_1,d_1}}_{\mbox{basis for }{\cal E}_{\pi^1,d_1}}|
 \ldots\ldots | \underbrace{\phi^m_{1,d_m}|\ldots |\phi^m_{d_m,d_m}}_{\mbox{basis for }{\cal E}_{\pi^m,d_m}}
 \Big].
\end{equation}
Here, we use $\phi^{k}_{i,j}$ to denote $\phi^{(\pi^k)}_{i,j}$ to ensure clear notation.
By the Schur orthogonality relations (Theorem~\ref{thm:schur-ortho}), $B$ is a unitary matrix, and block diagonalizes the right regular representation. Namely, for an arbitrary element $g\in \bbG$, the matrix $B^{-1} \rho(g) B$ is block diagonalized with $d_i$ blocks of size $d_i\times d_i$ for each irreducible representation $\pi^i\in\widehat{\bbG}$ as follows:
\begin{equation}\label{eq:explicit-diagonalization}
  B^{-1} \rho(g) B= 
\left[{\begin{array}{ccc}
\pi^1(g) & &0 \\
& \ddots& \\
0& & \pi^1(g)
\end{array}}\right]_{d_{1}\times d_{1}} \oplus\ldots\oplus
\left[\begin{array}{ccc}
\pi^m(g) & &0 \\
& \ddots& \\
0& & \pi^m(g)
\end{array}\right]_{d_{m}\times d_{m}}
\end{equation}
The importance of the above block diagonalization lies in the fact that the same unitary matrix $B$ simultaneously block diagonalizes all matrices $\rho(g)$ for every $g\in\bbG$.

We refer the reader to \cite{Serre} for a detailed account of the representation theory of finite groups.  We remark that many of the concepts we discussed in this section can be extended to the context of general compact groups; however, for the purposes of this article, we limit ourselves to finite groups. 

\subsection{Eigen-decompositions of (weighted) Cayley graph adjacency matrices} 
\label{subsec:eigen-decomp of A}
In this section, we review how the Frobenius--Schur Theorem can be used to find an explicit decomposition of adjacency matrices of weighted Cayley graphs.
We summarize this result in Proposition~\ref{prop:eigenvector}. While this result is well known, especially part $(i)$ in the case of unweighted Cayley graphs, it is challenging to find references that give detailed descriptions of the eigenbases mentioned in the proposition. In particular, we could not find references for proofs of Proposition \ref{prop:eigenvector} $(ii)$ and $(iii)$ in the case of a weighted Cayley graph that is not necessarily quasi-Abelian. To streamline the paper, we have included the proof of Proposition~\ref{prop:eigenvector} in  \ref{Appendix:proofs}.

We now provide the necessary definitions, \red{particularly of (weighted) Cayley graphs. Although Cayley graphs can be defined within the context of directed graphs, this paper focuses exclusively on undirected graphs.

Fix a finite group $\bbG$, and let $S$ be a subset of $\bbG$ that is closed under taking inverses, i.e., $s\in S$ precisely when $s^{-1}\in S$. We say $G$ is a \emph{Cayley graph} on the group $\bbG$ with \emph{generating set} $S$ if it has the vertex set $V(G)=\bbG$, and an edge joining vertex $g$ and $h$ whenever $g^{-1}h\in S$. }
We say $w:\bbG\rightarrow [0,\infty)$ is a \emph{weight function} if for each $x\in \bbG$, $w(x)=w(x^{-1})$. This leads to the following definition.

\begin{definition}\label{def-cayley}
    A \emph{weighted Cayley graph} $G$ on a group $\bbG$ with weight function $w:\bbG\rightarrow [0,\infty)$ has vertex set $V(G)=\bbG$ and edge weights determined by the weight function, where each edge $g\sim h$ has weight $w(g^{-1}h)$. If an edge has weight zero, we consider it as a non-edge in the weighted Cayley graph. We denote $G$ by ${\rm Cayley}(\bbG, w)$.
\end{definition}
 
Note that any Cayley graph is just a special case of a weighted Cayley graph where the weight function is ${\bf 1}_S$, the characteristic function of the generating set $S$ defined to be $${\bf 1}_S(x)=\begin{cases} 1 \text{ if } x\in S\\ 0 \text{ otherwise}\end{cases}.$$ Clearly, the adjacency matrix $A$ for a weighted Cayley graph ${\rm Cayley}(\bbG,w)$ as above has $(g,h)$-th entry 
$$A_{g,h}= w(g^{-1}h).$$ 

The following proposition relates the eigenvalues of the adjacency matrices of a weighted Cayley graph to eigenvalues of the corresponding irreducible matrix representations of the underlying group. 
To state our proposition, we need the following definition.
\begin{definition}\label{def:pi(S)}
Consider a function $f:\bbG\rightarrow \bbC$ where $\bbG$ is a finite group. For every representation $\pi\in\widehat{\bbG}$, define 
\begin{equation*}
    \pi(f)=\sum_{x\in \bbG}f(x)\pi(x).
\end{equation*}
\end{definition}
\begin{proposition}\label{prop:eigenvector}
Let $G$ be a weighted Cayley graph defined on a finite group $\bbG$ with weight function $w:\bbG\rightarrow [0,\infty)$.
\begin{itemize}
\item[(i)] The set of eigenvalues of the adjacency matrix $A_G$ coincides with the set $\bigcup_{\pi\in\widehat{\bbG}}{\rm spec}(\pi(w))$, where ${\rm spec}(\pi(w))$ denotes the spectrum of the matrix $\pi(w)$.
\item[(ii)] Let $B$ be the unitary matrix of the normalized coefficient functions given in \eqref{eq:matrix-B}. Every $\lambda$-eigenvector $\phi\in \bbC^{|\bbG|}$ of $A_G$ can be described as
$$\phi=B(\oplus_{\pi\in \widehat{\bbG}}\oplus_{i=1}^{d_\pi} X_{\pi,i}),$$
where every $X_{\pi,i}\in \bbC^{d_\pi}$ is either 0 or a $\lambda$-eigenvector for $\pi(w)$. Moreover, at least one of the vectors $X_{\pi,i}$ is nonzero. 
\item[(iii)] For each $\pi$, let ${\mathcal Q}_{\pi(w)}$ denote a fixed orthonormal eigenbasis for $\pi(w)$.
Then 
$$\bigcup_{\pi\in\widehat{\bbG}}\bigcup_{i:1,\ldots,d_\pi}\left\{\sqrt{\frac{d_\pi}{|\bbG|}}\sum_{k=1}^{d_\pi}x_k\pi_{k,i}: \ 
\left[
\begin{array}{c}
x_1 \\
\vdots \\
x_{d_\pi}  
\end{array}
\right]
\in {\mathcal Q}_{\pi(w)} \right\}$$
is an orthonormal eigenbasis for $A_G$. 
\end{itemize}
\end{proposition}
\red{The proof of Proposition~\ref{prop:eigenvector} can be found in  \ref{Appendix:proofs}.}
In the following corollary, we consider the special case for \emph{quasi-Abelian} Cayley graphs. A weighted Cayley graph is \emph{quasi-Abelian} if the weight function is a class function. A class function in $\ell^2(\bbG)$ (or the associated class vector in $\bbC^N$) is a function (or vector) that is constant on conjugacy classes of $\bbG$.
In this case, Proposition \ref{prop:eigenvector} takes a much simpler form. Namely, the collection of all (normalized) coefficient functions form an (orthonormal) eigenbasis for $A_G$.  Even though this result is known (see e.g.~\cite[Theorem 1.1]{Rockmore} or \cite[Theorem III.1]{MDK:2019:Sampta} for a proof), we provide a proof which is a simple application of the previous proposition. 
\begin{corollary}\label{cor:constant-on-conjugacy}
Consider a quasi-Abelian weighted Cayley graph $G$ on a finite group $\bbG$ of size $N$ with weight function $w:\bbG\rightarrow [0,\infty)$.
The set $\bigcup_{\pi\in \widehat{\bbG}}\left\{\sqrt{\frac{d_\pi}{N}}\pi_{i,j}:\ 1\leq i,j\leq d_\pi\right\}$ forms an orthonormal basis of eigenvectors for $A_G$. Namely, for every $\pi\in \widehat{\bbG}$ and $1\leq i,j\leq d_\pi$,
\[ A_G\pi_{i,j} = \lambda_\pi \pi_{i,j},\]
where $\lambda_\pi = \frac{1}{d_\pi}{\rm Tr}(\pi(w))$, \red{and $\pi(w)$ is as defined in Definition~\ref{def:pi(S)}}.
\end{corollary}
\begin{proof}
The set of normalized characters $\left\{\chi_\pi:=\frac{1}{\sqrt{N}}\sum_{i=1}^{d_{\pi}}\pi_{i,i}:\ \pi\in \widehat{\bbG}\right\}$ of a group $\bbG$ forms an orthonormal basis for the subspace of class functions in $\ell^2(\bbG)$ (see e.g.~\cite[Proposition 5.23]{1995:Folland:HarmonicAnalysis}).
Given that $G$ is quasi-Abelian, 
the weight function $w$ is a class function, so
\begin{equation*}\label{eq:expansion1}
    w=\sum_{\pi\in\widehat{\bbG}}\langle w, \chi_\pi\rangle_{\bbC^N}\chi_\pi
    =\frac{1}{\sqrt{N}}\sum_{\pi\in\widehat{\bbG}}\sum_{i=1}^{d_\pi}\langle w, \chi_\pi\rangle_{\bbC^N}\pi_{i,i}.
\end{equation*}
Next, we show that $\pi(w)$ is a multiple of the identity matrix. Indeed, using the above expansion of $w$ with respect to the orthonormal basis $\left\{\sqrt{\frac{d_{\pi}}{N} }\pi_{i,j}:\pi\in\widehat{\bbG}, 1\leq i,j\leq d_{\pi} \right\}$, we have that for every $\pi\in\widehat{\bbG}$,
\begin{eqnarray*}
\langle \pi(w)e_i,e_j\rangle = \sum_{x\in \bbG}\langle w(x)\pi(x)e_i,e_j\rangle= \sum_{x\in \bbG}w(x)\pi_{i,j}(x)= 
\overline{\langle w, \pi_{i,j}\rangle_{\bbC^N}}
=\left\{\begin{array}{cc}
  0   & i\neq j \\
  \frac{\sqrt{N}}{d_\pi}\langle\chi_\pi, w\rangle_{{\bbC^N}}   & i=j
\end{array}\right..
\end{eqnarray*}
So, the standard basis $\{e_i\}_{i=1}^{d_{\pi}}$ is an orthonormal eigenbasis of $\pi(w)$ associated with the eigenvalue 
$\frac{\sqrt{N}}{d_\pi}\langle\chi_\pi, w\rangle_{{\bbC^N}}$. By part $(iii)$ of Proposition \ref{prop:eigenvector}, the set $\bigcup_{\pi\in\widehat{\bbG}}\{\pi_{i,j}: 1\leq i,j\leq d_\pi\}$ forms an orthogonal eigenbasis for $A_{G}$ associated with (repeated) eigenvalues $\frac{\sqrt{N}}{d_\pi}\langle\chi_\pi, w\rangle_{{\bbC^N}}$. Finally, since for every $i$ we have
$\frac{\sqrt{N}}{d_\pi}\langle\chi_\pi, w\rangle_{{\bbC^N}}=\langle\pi(w) e_{i}, e_i\rangle_{{\bbC^N}},$
we get $\frac{\sqrt{N}}{d_\pi}\langle\chi_\pi, w\rangle_{{\bbC^N}}=\frac{1}{d_\pi}{\rm Tr}(\pi(w))$. 
\end{proof}

Next, we observe the inherent difference between the spectral behaviour of general weighted Cayley graphs and quasi-Abelian ones. As shown in Corollary~\ref{cor:constant-on-conjugacy}, the canonical eigenbasis for the adjacency matrix of a quasi-Abelian weighted Cayley graph does not depend on the particular choice of the weight function $w$. The choice of the weight function only affects the eigenvalues. In the case of a general weighted Cayley graph (as in Proposition~\ref{prop:eigenvector}), the weight function determines both eigenvalues and eigenvectors of the associated Cayley graph. As a result, the spectral features of non-quasi-Abelian weighted Cayley graphs are more intricate.

We conclude this section with the illustration of Proposition~\ref{prop:eigenvector} through two examples.
Due to applications in analysis of ranked data sets, developing signal processing on the symmetric group is of considerable interest. 
In \cite{paper}, the authors contend that the permutahedron is an appropriate Cayley graph on $\bbS_n$ for representing ranked data.
In this graph, two rankings are  adjacent if and only if they differ by transposing two adjacent candidates.
More precisely, 
the \emph{permutahedron} ${\mathbb P}_n$ is the Cayley graph on the group of $\bbS_n$ with generating set $S_n=\{(i,i+1):\ 1\leq i\leq n-1\}$ consisting of all consecutive transpositions.

\begin{example}[Fourier basis for signal processing on ${\mathbb P}_3$]\label{example:P3}
Using the notation from \ref{Appendix:rep-S3} for the irreducible representations of $\bbS_3$, define
$B=\begin{bmatrix} \frac{1}{\sqrt{6}}\iota_{1,1} | \frac{1}{\sqrt{6}}\tau_{1,1} | \frac{1}{\sqrt{3}}\pi_{1,1}|\frac{1}{\sqrt{3}}\pi_{2,1}|\frac{1}{\sqrt{3}}\pi_{1,2}|\frac{1}{\sqrt{3}}\pi_{2,2}\end{bmatrix}$, \red{where $\iota$, $\tau$, and $\pi$ are the trivial representation, the alternating representation, and the standard representation of $\bbS_3$, respectively}. 
\begin{figure}[h]
  \centering
   \subfigure{
\begin{tikzpicture}[scale=0.4, baseline={(0,0)}]
\begin{scope}[every node/.style={circle,thick,draw}]
    \node (A) at (0,4) [label=above:(13)]{}; 
    \node (B) at (3.464,2) [label=right:(132)]{}; 
    \node (C) at (3.464,-2) [label=right:(23)]{}; 
    \node (D) at (0,-4) [label=below:id]{}; 
    \node (E) at (-3.464,-2) [label=left:(12)]{}; 
    \node (F) at (-3.464,2) [label=left:(123)]{}; 
\end{scope}

\begin{scope}[>={Stealth[black]},
              every node/.style={fill=white,circle},
              every edge/.style={draw=red,very thick}]
    \draw[thick, purple] (D)--(C);
    \draw[thick, teal] (D)--(E);
    \draw[thick, purple] (E)--(F);
    \draw[thick, teal] (C)--(B);
    \draw[thick, purple] (B)--(A);
    \draw[thick, teal] (F)--(A);
\end{scope}
\end{tikzpicture}
  }
 \qquad
  \subfigure{
    \begin{tabular}{c}
      $A_3=\begin{bmatrix}
      0& 1& 1& 0& 0& 0\\
1& 0& 0& 0& 1& 0\\ 
1& 0& 0& 0& 0& 1\\ 
0& 0& 0& 0& 1& 1\\ 
0& 1& 0& 1& 0& 0\\ 
0& 0& 1& 1& 0& 0
      \end{bmatrix}$
    \end{tabular}
  }
  \caption{\red{The Cayley graph ${\mathbb P}_3$ and its adjacency matrix $A_3$.}} 
  \label{P3}
\end{figure}

The unitary matrix $B$ block-diagonalizes matrices $\rho(g)$ for all $g\in \bbS_3$. So, for the adjacency matrix $A_3$ of ${\mathbb P}_3$ we have
$$B^{\t} A_3B={\rm diag}\left(2,-2,\begin{bmatrix} 
\frac{1}{2} & \frac{\sqrt{3}}{2} \\
\frac{\sqrt{3}}{2} & -\frac{1}{2}
\end{bmatrix},\begin{bmatrix} 
\frac{1}{2} & \frac{\sqrt{3}}{2} \\
\frac{\sqrt{3}}{2} & -\frac{1}{2}
\end{bmatrix}\right).$$
The eigenvalues of $B^{\t}A_3B$ are $2,-2\ (\mbox{of multiplicity 1}), 1, -1\ (\mbox{of multiplicity 2})$ associated with eigenvectors
$[1,0,0,0,0,0]^{\t},[0,1,0,0,0,0]^{\t},[0,0,\sqrt{3}/2,1/2,0,0]^{\t},$ $[0,0,0,0,\sqrt{3}/2,1/2]^{\t},
[0,0,-1/2,\sqrt{3}/2,0,0]^{\t}$, $[0,0,0,0,-1/2,\sqrt{3}/2]^{\t}$.
This leads to the following \red{eigenbasis for $A_3$, derived from the irreducible representations of $\bbS_3$:}
$$
{\mathcal B}=\left\{\begin{bmatrix}
\frac{1}{\sqrt{6}}\\
\frac{1}{\sqrt{6}}\\
\frac{1}{\sqrt{6}}\\
\frac{1}{\sqrt{6}}\\
\frac{1}{\sqrt{6}}\\
\frac{1}{\sqrt{6}}
\end{bmatrix},
\begin{bmatrix}
\frac{1}{\sqrt{6}}\\
-\frac{1}{\sqrt{6}}\\
-\frac{1}{\sqrt{6}}\\
-\frac{1}{\sqrt{6}}\\
\frac{1}{\sqrt{6}}\\
\frac{1}{\sqrt{6}}
\end{bmatrix},
\begin{bmatrix}
\frac12\\
0\\
\frac12\\
-\frac12\\
-\frac12\\
0
\end{bmatrix},
\begin{bmatrix}
\frac{1}{2\sqrt{3}}\\
\frac{1}{\sqrt{3}}\\
-\frac{1}{2\sqrt{3}}\\
-\frac{1}{2\sqrt{3}}\\
\frac{1}{2\sqrt{3}}\\
-\frac{1}{\sqrt{3}}
\end{bmatrix},
\begin{bmatrix}
-\frac{1}{2\sqrt{3}}\\
\frac{1}{\sqrt{3}}\\
-\frac{1}{2\sqrt{3}}\\
-\frac{1}{2\sqrt{3}}\\
-\frac{1}{2\sqrt{3}}\\
\frac{1}{\sqrt{3}}
\end{bmatrix},
\begin{bmatrix}
\frac{1}{2}\\
0\\
-\frac{1}{2}\\
\frac{1}{2}\\
-\frac{1}{2}\\
0
\end{bmatrix}\right\}.
$$
\end{example}

\begin{remark}\label{remark:compare-with-cycle}
    \red{The permutahedron $\mathbb P_3$ is pictured in Figure \ref{P3}, and is isomorphic to the 6-cycle. The Cayley graph of $\bbZ_6$ with generating set $\{1,-1\}$ is isomorphic to the 6-cycle as well. We note that the underlying groups of these two Cayley graphs are very different; in fact, one group is Abelian, and the other is not. 
    While these graphs are isomorphic, taking each underlying group and its representation theory into account results in producing different Fourier bases.} 
    For the Cayley graph of $\bbZ_6$ with generating set $\{1,-1\}$, the Fourier basis consists of the normalized group characters for $\bbZ_6$,
          $$
    \left\{\begin{bmatrix}
        \frac{1}{\sqrt{6}} \\ \frac{1}{\sqrt{6}} \\ \frac{1}{\sqrt{6}} \\ \frac{1}{\sqrt{6}} \\ \frac{1}{\sqrt{6}} \\ \frac{1}{\sqrt{6}}
    \end{bmatrix},
    \begin{bmatrix}
        \frac{1}{\sqrt{6}} \\ \frac{1}{\sqrt{6}}\omega \\ \frac{1}{\sqrt{6}}\omega^2 \\ \frac{1}{\sqrt{6}}\omega^3 \\ \frac{1}{\sqrt{6}}\omega^4 \\ \frac{1}{\sqrt{6}}\omega^5
    \end{bmatrix},
    \begin{bmatrix}
        \frac{1}{\sqrt{6}} \\ \frac{1}{\sqrt{6}}\omega^2 \\ \frac{1}{\sqrt{6}}\omega^4 \\ \frac{1}{\sqrt{6}} \\ \frac{1}{\sqrt{6}}\omega^2 \\ \frac{1}{\sqrt{6}}\omega^4
    \end{bmatrix},
    \begin{bmatrix}
        \frac{1}{\sqrt{6}} \\ \frac{1}{\sqrt{6}}\omega^3 \\ \frac{1}{\sqrt{6}} \\ \frac{1}{\sqrt{6}}\omega^3 \\ \frac{1}{\sqrt{6}} \\ \frac{1}{\sqrt{6}}\omega^3
    \end{bmatrix},
    \begin{bmatrix}
        \frac{1}{\sqrt{6}} \\ \frac{1}{\sqrt{6}}\omega^4 \\ \frac{1}{\sqrt{6}}\omega^2 \\ \frac{1}{\sqrt{6}} \\ \frac{1}{\sqrt{6}}\omega^4 \\ \frac{1}{\sqrt{6}}\omega^2
    \end{bmatrix},
    \begin{bmatrix}
        \frac{1}{\sqrt{6}} \\ \frac{1}{\sqrt{6}}\omega^5 \\ \frac{1}{\sqrt{6}}\omega^4 \\ \frac{1}{\sqrt{6}}\omega^3 \\ \frac{1}{\sqrt{6}}\omega^2 \\ \frac{1}{\sqrt{6}}\omega
    \end{bmatrix}\right\},
    $$
     where $\omega=\exp(2\pi i/6)$ is the primitive sixth root of unity. \red{Similar graph Fourier bases, where the underlying group is taken to be $\bbZ_n$, have been used in \cite{GSP-circulant2019} for developing GSP methods applicable to circulant networks.}

\red{On the other hand, the Fourier basis for $\mathbb P_3$ is $\mathcal B$, as described in Example~\ref{example:P3}. A Fourier basis similar to ${\mathcal B}$, where the underlying group is taken to be the symmetric group, was used in \cite{paper} to develop GSP methods for analyzing ranked data sets.
The symmetric group $\mathbb{S}_3$, which consists of all permutations of three objects, and its Cayley graphs are particularly well-suited for developing GSP when the data consists of rankings involving three items (see \cite{paper} for applications of GSP based on ${\mathbb P}_n$ to analysis of ranked data).}
\end{remark}

\begin{example}[Fourier basis for signal processing on ${\mathbb P}_4$]\label{example:P4}
Using the notation for the irreducible representations of $\bbS_4$ given in \ref{Appendix:rep-S4}, we define the matrix 
$$B=\left[\begin{smallmatrix} \widetilde{\iota_{1,1}} | \widetilde{\tau_{1,1}} | \widetilde{\sigma_{1,1}}|\widetilde{\sigma_{2,1}}|\widetilde{\sigma_{1,2}}|\widetilde{\sigma_{2,2}}|  \widetilde{\pi_{1,1}}|\widetilde{\pi_{2,1}}|\widetilde{\pi_{3,1}}|\widetilde{\pi_{1,2}}|\widetilde{\pi_{2,2}}|\widetilde{\pi_{3,2}}|\widetilde{\pi_{1,3}}|\widetilde{\pi_{2,3}}|\widetilde{\pi_{3,3}}| \widetilde{\pi'_{1,1}}|\widetilde{\pi'_{2,1}}|\widetilde{\pi'_{3,1}}|\widetilde{\pi'_{1,2}}|\widetilde{\pi'_{2,2}}|\widetilde{\pi'_{3,2}}|\widetilde{\pi'_{1,3}}|\widetilde{\pi'_{2,3}}|\widetilde{\pi'_{3,3}}\end{smallmatrix}\right],$$
where $\widetilde{\upsilon_{i,j}}$ denotes the normalization of the coefficient function $\upsilon_{i,j}$ \red{into a unit vector} for every representation $\upsilon$.
\red{Here, $\iota$, $\tau$, $\sigma$, $\pi'$, and $\pi$ are the trivial representation, the alternating representation, the 2-dimensional representation, the standard representation, and the standard-alternating tensor representation of $\bbS_4$, respectively.}
Note that $B$ is a unitary matrix that block diagonalizes the right regular representation of $\bbS_4$.
%
\begin{figure}[h]\label{fig:P_4}   
\centering
\begin{tikzpicture}
    \begin{scope}[scale=0.8, every node/.style={circle,thick,draw}]
        \node (id) at (0,0) [label=below:id]{};
       
        \node (2134) at (-6,1) [label=below:(12)]{};
        \node (1324) at (0,1) [label=left:(23)]{};
        \node (1243) at (6,1) [label=below:(34)]{};

        \node (2314) at (-7,2) [label=left:(123)]{};
        \node (3124) at (-3,2) [label=180:(132)]{};
        \node (2143) at (0,2) [label=2:(12)(34)]{};
        \node (1342) at (3,2) [label=right:(234)]{};
        \node (1423) at (7,2) [label=right:(243)]{};

        \node (3214) at (-8,3) [label=left:(13)]{};
        \node (2341) at (-5,3) [label=right:(1234)]{};
        \node (2413) at (-2,3) [label=right:(1243)]{};
        \node (3142) at (2,3) [label=right:(1342)]{};
        \node (1432) at (5,3) [label=right:(24)]{};
        \node (4123) at (8,3) [label=right:(1432)]{};

        \node (3241) at (-7,4) [label=left:(134)]{};
        \node (2431) at (-3,4) [label=left:(124)]{};
        \node (3412) at (0,4) [label=2:(13)(24)]{};
        \node (4213) at (3,4) [label=0:(143)]{};
        \node (4132) at (7,4) [label=right:(142)]{};

        \node (3421) at (-6,5) [label=above:(1324)]{};
        \node (4231) at (0,5) [label=left:(14)]{};
        \node (4312) at (6,5) [label=above:(1423)]{};

        \node (4321) at (0,6) [label=above:(14)(23)]{};
    \end{scope}

    \begin{scope}[>={Stealth[black]},
              every node/.style={fill=white,circle},
              every edge/.style={draw=red,very thick}]
        \draw[thick,teal] (id)--(2134);
        \draw[thick,purple] (id)--(1324);
        \draw[thick,orange] (id)--(1243);

        \draw[thick,purple] (2134)--(2314);
        \draw[thick,orange] (2134)--(2143);
        \draw[thick,teal] (1324)--(3124);
        \draw[thick,orange] (1324)--(1342);
        \draw[thick,teal] (1243)--(2143);
        \draw[thick,purple] (1243)--(1423);

        \draw[thick,teal] (2314)--(3214);
        \draw[thick,orange] (2314)--(2341);
        \draw[thick,purple] (2143)--(2413);
        \draw[thick,purple] (3124)--(3214);
        \draw[thick,orange] (3124)--(3142);
        \draw[thick,teal] (1342)--(3142);
        \draw[thick,purple] (1342)--(1432);
        \draw[thick,orange] (1423)--(1432);
        \draw[thick,teal] (1423)--(4123);

        \draw[thick,orange] (3214)--(3241);
        \draw[thick,teal] (2341)--(3241);
        \draw[thick,purple] (2341)--(2431);
        \draw[thick,orange] (2413)--(2431);
        \draw[thick,teal] (2413)--(4213);
        \draw[thick,purple] (3142)--(3412);
        \draw[thick,teal] (1432)--(4132);
        \draw[thick,purple] (4123)--(4213);
        \draw[thick,orange] (4123)--(4132);

        \draw[thick,purple] (3241)--(3421);
        \draw[thick,teal] (2431)--(4231);
        \draw[thick,orange] (4213)--(4231);
        \draw[thick,orange] (3412)--(3421);
        \draw[thick,teal] (3412)--(4312);
        \draw[thick,purple] (4132)--(4312);

        \draw[thick,teal] (3421)--(4321);
        \draw[thick,purple] (4231)--(4321);
        \draw[thick,orange] (4312)--(4321);
    \end{scope}
    \end{tikzpicture}
\captionof{figure}{The permutahedron $\mathbb P_4$ is the Cayley graph on the group $\bbS_4$ with generating set $S=\{(12), (23), (34)\}$.}
\end{figure}

So, for the adjacency matrix $A_4$ of ${\mathbb P}_4$ we have
$$B^{-1} A_4B={\rm diag}\left(3,\ -3,\ 
\bigoplus_{i=1}^2\begin{bmatrix} 
0 & 2+\omega^3 \\
2+\omega &0
\end{bmatrix},\ 
\bigoplus_{i=1}^3\begin{bmatrix} 
-2 & 0 & -1 \\
0 & -1 & 0 \\
-1 & 0 & 0
\end{bmatrix},\ 
\bigoplus_{i=1}^3\begin{bmatrix} 
2 & 0 & 1 \\
0 & 1 & 0 \\
1 & 0 & 0
\end{bmatrix}\right).$$
Next we apply Proposition~\ref{prop:eigenvector} $(ii)$ to the eigenvectors of the blocks in $B^{-1}A_4B$
to compute eigenvectors of $A_4$ as follows. 
\begin{align*}
{\cal B}'=\left\{\begin{matrix}
    \widetilde{\iota_{1,1}}, \widetilde{\tau_{1,1}},\frac{\sqrt{6}-\sqrt{2}i}{4}\widetilde{\sigma_{1,l}}+\frac{\sqrt{2}}{2}\widetilde{\sigma_{2,l}}, \frac{-\sqrt{6}+\sqrt{2}i}{4}\widetilde{\sigma_{1,l}}+\frac{\sqrt{2}}{2}\widetilde{\sigma_{2,l}},& &  \\
    \frac{1-\sqrt{2}}{\sqrt{4-2\sqrt{2}}}\widetilde{\pi_{1,k}}+\frac{1}{\sqrt{4-2\sqrt{2}}}\widetilde{\pi_{3,k}}, 
    \widetilde{\pi_{2,k}}, 
    \frac{\sqrt{2}+1}{\sqrt{4+2\sqrt{2}}}\widetilde{\pi_{1,k}}+\frac{1}{\sqrt{4+2\sqrt{2}}}\widetilde{\pi_{3,k}}, &:& k=1,2,3, \ l=1,2\\
    \frac{\sqrt{2}+1}{\sqrt{4+2\sqrt{2}}}\widetilde{\pi_{1,k}'}+\frac{1}{\sqrt{4+2\sqrt{2}}}\widetilde{\pi_{3,k}'}, 
    \widetilde{\pi_{2,k}'}, 
    \frac{1-\sqrt{2}}{\sqrt{4-2\sqrt{2}}}\widetilde{\pi_{1,k}'}+\frac{1}{\sqrt{4-2\sqrt{2}}}\widetilde{\pi_{3,k}'} & &  
\end{matrix}\right\}.
\end{align*}

\end{example}

\section{Construction of frames for graph signals}\label{sec:frames}
Let $\bbG$ be a finite (not necessarily Abelian) group, and consider the associated inner product space $\ell^2(\bbG)$ together with its usual norm $\|f\|_2=\sqrt{\sum_{x\in \bbG} |f(x)|^2}$. 
 A \emph{frame} for $\ell^2(\bbG)$ is a finite set of vectors ${\mathcal F}=\{\psi_i: i\in I\}$ such that for some positive real numbers $A$ and $B$, we have 
\begin{equation}\label{Eframedefn}
  A\|f\|_{2}^2
    \leq \sum_{i\in I} |\langle f, \psi_i \rangle |^2
    \leq B\|f\|_{2}^2, \mbox{ for every }f\in \ell^2(\bbG).%
\end{equation}%
\red{Elements $\psi_i$ of a frame ${\mathcal F}$ are called \emph{frame atoms}.}
The constants $A$ and $B$ are called the \emph{lower} and \emph{upper frame bounds} respectively. 
 Frames provide stable and possibly redundant systems which allow reconstruction of a signal $f$ from its frame coefficients $\{\langle f, \psi_i \rangle\}_{i\in I}$. In the case of redundant frames, reconstruction of a signal might still be possible even if some portion of its frame coefficients is lost or corrupted.

Since $\ell^2(\bbG)$ is finite-dimensional, any finite spanning subset of it forms a frame. Frames, however, differ significantly from each other in terms of how efficiently they analyze signals.
The \emph{condition number} of a frame $\mathcal{F}$ is defined to be the ratio $c(\mathcal{F}) := B/A$, where $A, B$ denote the optimal constants satisfying Equation \eqref{Eframedefn}.  
An important class of frames is the class of \emph{tight frames}, i.e., frames for which $A=B$. \emph{Parseval frames} are tight frames in which $A=B=1$. Compared to general frames or to orthonormal bases, tight frames exhibit many desirable properties, such as greater numerical stability when reconstructing noisy signals. A major goal in designing frames for real applications is to design tight frames, or at least frames with a small condition number. For a detailed introduction to frame theory, see \cite{Christensen:2013:IntroductionToFrames, HanLarson:FramesIntro:2000}.

\subsection{Special frames for Cayley graphs}\label{subsec:FS-frames-Cayley-frames}
Throughout this section, let $G$ be the weighted Cayley graph of the finite group $\bbG$ with a given weight function $w:\bbG\rightarrow [0,\infty)$. We  construct  frames for $\ell^2(\bbG)$ that are compatible with the weight function $w$ and the representation theory of $\bbG$. 
Recall that for a representation $\pi$ of $\bbG$ and a weight function $w:\bbG\rightarrow [0,\infty)$, we set $\pi(w)=\sum_{x\in \bbG}w(x)\pi(x)$; and that the matrix $\pi(w)$ is self-adjoint.
\begin{definition}\label{def:frame-compatible}
Let $\bbG$ be a finite group. 
\begin{itemize}
    \item[(i)] A frame $\{\psi_i\}_{i=1}^m$  for $\ell^2(\bbG)$ is called a Frobenius--Schur frame if each atom $\psi_i$ belongs to one orthogonal component of the Frobenius--Schur decomposition as stated in Theorem~\ref{thm:frob-schur}. \red{More precisely, every frame atom $\psi_i$ belongs to some $\mathcal{E}_{\pi,j}$ (as defined in \eqref{eq:def-E-pi,i}) for an irreducible representation $\pi$ of $\widehat{\bbG}$ and $1\leq j\leq d_\pi$.}
    \item[(ii)] Let $\pi$ be an irreducible representation of ${\bbG}$ of dimension $d_\pi$, and let $1\leq i\leq d_\pi$ be fixed. 
    For $\lambda\in {\mathbb R}$, define 
\begin{equation}\label{eq:Z-space}
        Z_{\pi, i, \lambda}:=\left\{\sum_{k=1}^{d_\pi}x_k\pi_{k,i}: \ 
        \left[
        \begin{array}{c}
        x_1 \\
        \vdots \\
        x_{d_\pi}  
        \end{array}
        \right]
        \in E_{\lambda}(\pi(w)) \right\},
    \end{equation}
    where $E_{\lambda}(\pi(w))$ is the $\lambda$-eigenspace of $\pi(w)$. Here, we use the convention that if $\lambda\not\in{\rm spec}(\pi(w))$, then $E_{\lambda}(\pi(w))=\{{ 0}\}$.
    \item[(iii)] A  frame $\{\psi_i\}_{i=1}^m$ of $\ell^2(\bbG)$ is said to be  a \emph{$w$-Cayley frame} (or \emph{Cayley frame compatible with $w$}) if for every atom $\psi_i$ there exists an irreducible representation $\pi:\bbG\to {\mathcal U}_{d_\pi}(\bbC)$, an eigenvalue $\lambda\in{\rm spec}(\pi(w))$ and an index $1\leq j\leq d_\pi$ such that $\psi_i\in Z_{\pi,j,\lambda}$. 
\end{itemize}
\end{definition}

\begin{remark}\label{rem:theta}
\begin{itemize}
    \item[(i)] Any $w$-Cayley frame of $\ell^2(\bbG)$ is in fact a Frobenius--Schur frame. Indeed, for every $k=1,\ldots, d_\pi$, the vector $\pi_{k,i}$ belongs to $\E_{\pi,i}$, which implies that $Z_{\pi,i,\lambda}\subseteq \E_{\pi,i}$. 
    \item[(ii)] 
    Fix $i=1,\ldots,d_\pi$.
    Let $\{e_j\}_{j=1}^{d_\pi}$ denote the standard orthonormal basis of $\bbC^{d_\pi}$, and recall that $\left\{\sqrt{\frac{d_{\pi}}{|\bbG|}}\pi_{j,i}: 1\leq j\leq d_\pi\right\}$ forms an orthonormal basis for $\E_{\pi,i}$.
    \red{For $\pi$ and $i$ as above, define a linear map as follows:}
    \begin{equation}\label{eq:isom-theta}
    \Theta_{\pi,i}:\bbC^{d_\pi}\to \E_{\pi,i}, \quad
\Theta_{\pi, i}([x_1 \ldots x_{d_{\pi}}]^{\t})= \sum_{k=1}^{d_\pi}x_k\sqrt{\frac{d_{\pi}}{|\bbG|}}\pi_{k,i}.
\end{equation}
\red{Each map $\Theta_{\pi,i}$ is an isometric isomorphism, as it maps one orthonormal basis to another.}
    For every $\lambda\in {\rm spec}(\pi(w))$, we define the restriction map $\Theta_{\pi,i,\lambda}:=\Theta_{\pi,i}|_{E_\lambda(\pi(w))}$. This map forms an isometric isomorphism between $E_\lambda(\pi(w))$ and 
    $Z_{\pi, i, \lambda}$.
\end{itemize} 
\end{remark}

\begin{proposition}\label{prop:Z-pi-decomp}
For every $\pi\in\widehat{\bbG}$ and $1\leq i\leq d_\pi$, we have $\E_{\pi,i}=\bigoplus_{\lambda\in {\rm spec}(\pi(w))} Z_{\pi,i,\lambda}.$ 
\end{proposition}
\begin{proof}
The matrix $\pi(w)$ is  self-adjoint (i.e., $\pi(w)^*=\pi(w)$), and as a result it is diagonalizable. So, one can build an orthonormal basis of $\bbC^{d_\pi}$ consisting of eigenvectors of $\pi(w)$. In other words, we can write 
\begin{equation}\label{eq:spectral-decomp}
    \bbC^{d_\pi}=\bigoplus_{\lambda\in{\rm spec}(\pi(w))} E_\lambda(\pi(w)). 
\end{equation}
Now, consider an arbitrary element $\pi_{\xi,i}\in\E_{\pi,i}$, and write its linear expansion $\pi_{\xi,i}=\sum_{k=1}^{d_\pi} x_k\pi_{k,i}$. Using \eqref{eq:spectral-decomp}, the vector 
$X=[x_1,\ldots,x_{d_\pi}]^{\t}\in \bbC^{d_\pi}$ can be written as a linear combination $X=\sum_{\lambda\in {\rm spec}(\pi(w))} Y_\lambda$, with $Y_\lambda\in E_\lambda(\pi(w))$. 
Letting $Y_\lambda=[y_{\lambda,1},\ldots,y_{\lambda,d_\pi}]^{\t}$, we have 
$$\pi_{\xi,i}=[\pi_{1,i}|\ldots|\pi_{d_\pi,i}][x_1,\ldots,x_{d_\pi}]^{\t}=\sum_{\lambda\in {\rm spec}(\pi(w))} [\pi_{1,i}|\ldots|\pi_{d_\pi,i}]Y_\lambda=\sum_{\lambda\in {\rm spec}(\pi(w))}\sum_{k=1}^{d_\pi} y_{\lambda,k}\pi_{k,i}.$$
Note that for each $\lambda$, we have $\sum_{k=1}^{d_\pi} y_{\lambda,k}\pi_{k,i}$ belongs to $Z_{\pi,i,\lambda}$. So,  
we get
$\E_{\pi,i}\subseteq \sum_{\lambda\in {\rm spec}(\pi(w))} Z_{\pi,i,\lambda}.$ 
On the other hand,  $\E_{\pi,i}\supseteq \sum_{\lambda\in {\rm spec}(\pi(w))} Z_{\pi,i,\lambda}$ is a trivial consequence of the definition of $Z_{\pi, i, \lambda}$. 

Lastly we show that this sum is a direct sum.  
That is, $\{Z_{\pi,i,\lambda} :\lambda\in {\rm spec }(\pi(w))\}$ is a set of orthogonal subspaces. To prove this, consider $Z_{\pi,i,\lambda}$ and $Z_{\pi,i,\lambda'}$ for distinct eigenvalues $\lambda,\lambda'\in {\rm spec}(\pi(w))$.
Take arbitrary $X\in E_{\lambda}(\pi(w))$ and $Y\in E_{\lambda'}(\pi(w))$.
Since $E_\lambda(\pi(w))$ and $E_{\lambda'}(\pi(w))$ are orthogonal subspaces of $\bbC^{d_\pi}$, 
we have $\langle X,Y\rangle_{{\mathbb C}^{d_\pi}}=0$. Now using Schur's orthogonality relations (Theorem~\ref{thm:schur-ortho}), we get
\begin{equation*}
    \left\langle \sum_{k=1}^{d_\pi}x_k\pi_{k,i} , \sum_{\ell=1}^{d_\pi}y_\ell\pi_{\ell,i}\right\rangle_{\ell^2(\bbG)}=\sum_{k,\ell=1}^{d_\pi} x_k\overline{y_\ell}\langle\pi_{k,i} ,\pi_{\ell,i}\rangle_{\ell^2(\bbG)}
    =\sum_{k=1}^{d_\pi} x_k\overline{y_k}\frac{|\bbG|}{d_\pi}=\frac{|\bbG|}{d_\pi}\langle X,Y\rangle_{{\mathbb C}^{d_\pi}}=0.
\end{equation*}
This completes the proof.
\end{proof}

Now, we can provide a characterization for all
$w$-Cayley frames of $\ell^2(\bbG)$. To state our theorem, we need the following notation.

\begin{notation}
For $\pi\in \widehat{\bbG}$ and an eigenvalue $\lambda$ of $\pi(w)$, let ${\mathfrak G}_{\pi,\lambda}$ denote the collection of all frames for  the eigenspace $E_\lambda(\pi(w))$. We define ${\mathfrak G}_{\pi,\lambda}=\emptyset$ if $\lambda$ is not an eigenvalue of $\pi(w)$.
Elements of ${\mathfrak G}_{\pi,\lambda}$ are denoted by calligraphic font, e.g.~${\cal F}$.
\end{notation}

\begin{theorem}\label{thm: frame} 
For every $\pi\in \widehat{\bbG}$ of dimension $d_\pi$, every eigenvalue $\lambda\in {\rm spec}(\pi(w))$, and every index $1\leq i\leq d_\pi$, let 
${\cal F}^{\pi,\lambda}_i\in {\mathfrak G}_{\pi,\lambda}$ be a given frame with lower and upper frame bounds $A^{\pi,\lambda}_i$ and $B^{\pi,\lambda}_i$, respectively. 
Then,
\begin{itemize}
\item[(i)] The collection 
$${\mathcal G}_{\pi,i}=\left\{\Theta_{\pi,i,\lambda}(\psi):\  \psi\in {\cal F}^{\pi,\lambda}_i, \,  \lambda\in{\rm spec}(\pi(w))\right\}$$ is a frame for $\E_{\pi,i}$, with lower frame bound $A_{\pi,i}=\min_{\lambda\in {\rm spec}(\pi(w))}A_i^{\pi,\lambda}$ and upper frame bound $B_{\pi,i}=\max_{\lambda\in {\rm spec}(\pi(w))} B_i^{\pi,\lambda}$.
\item[(ii)] The collection 
${\mathcal{G}}=\bigcup_{\pi\in\widehat{\bbG}, 1\leq i\leq d_\pi}{\mathcal G}_{\pi,i}$
is a Frobenius--Schur frame for $\ell^2(\bbG)$, which is also a $w$-Cayley frame.\footnote{\red{By Remark~\ref{rem:theta} (i), $w$-Cayley frames are a special type of Frobenius--Schur frames.}} Moreover, ${\mathcal G}$ has lower frame bound $\min\{A_{\pi,i}:\pi\in \widehat{\bbG}, 1\leq i\leq d_{\pi}\}$ and upper frame bound $\max\{B_{\pi,i}:\pi\in \widehat{\bbG}, 1\leq i\leq d_{\pi}\}$.
\item[(iii)] Every  $w$-Cayley frame for $\ell^2(\bbG)$ is of the form described in $(ii)$.
\end{itemize}
\end{theorem}
\begin{proof}
Fix $\pi\in \widehat{\bbG}$, and $1\leq i\leq d_{\pi}$. Let $\lambda\in {\rm spec}(\pi(w))$. Then if the frame ${\cal F}_{i}^{\pi,\lambda}$ has lower and upper frame bounds $A_{i}^{\pi,\lambda}$ and $B_{i}^{\pi,\lambda}$ respectively, then $\widetilde{\cal F}_i^{\pi,\lambda}=\{{\Theta}_{\pi,i,\lambda}(\psi):\psi\in {\cal F}_{i}^{\pi,\lambda}\}$ forms a frame for $Z_{\pi,i,\lambda}$ with the same lower and upper frame bounds. This is indeed the case, as  ${\Theta}_{\pi,i,\lambda}$ is an isometric isomorphism by Remark~\ref{rem:theta}. 
Using the direct-sum decomposition of Proposition~\ref{prop:Z-pi-decomp}, this union results in a frame for $\E_{\pi,i}$ whose  lower and upper frame bounds are the minimum and maximum, respectively, over all frame bounds for the frames $\widetilde{\cal F}_i^{\pi,\lambda}$. 
To see this,  consider an arbitrary $f\in \E_{\pi,i}$, and note that $f=\sum_{\lambda\in {\rm spec}(\pi(w))}f_{\lambda}$ for $f_\lambda\in Z_{\pi,i,\lambda}$. Fixing $\lambda\in {\rm spec}(\pi(w))$, we have that 
\begin{equation}\label{proof:frame-i}
    A_{i}^{\pi,\lambda} \|f_{\lambda}\|_2^2\leq \sum_{\phi\in \widetilde{\cal F}_i^{\pi,\lambda}} |\langle f_{\lambda}, \phi \rangle|^2\leq B_{i}^{\pi,\lambda} \|f_{\lambda}\|_2^2
\end{equation}
since $\widetilde{\cal F}_{i}^{\pi,\lambda}$ is a frame for $Z_{\pi,i,\lambda}$. Equation (\ref{proof:frame-i}) holds for all $\lambda\in {\rm spec}(\pi(w))$, and we have that $\sum_{\lambda\in{\rm spec}(\pi(w))} \|f_{\lambda}\|_2^2=\|f\|_2^2$ since the decomposition of $\E_{\pi,i}$ is orthogonal. Summing over $\lambda$ gives
\begin{align*}
    \sum_{\lambda\in {\rm spec}(\pi(w))} \sum_{\psi_\alpha\in \widetilde{
    {\cal F}}_i^{\pi,\lambda}} |\langle f_{\lambda}, \psi_\alpha \rangle|^2 &\geq  \sum_{\lambda\in {\rm spec}(\pi(w))} A_{i}^{\pi,\lambda} \|f_\lambda\|_2^2 
    \geq \left(\min_{\lambda\in{\rm spec}(\pi(w))}A_{i}^{\pi,\lambda}\right) \|f\|_2^2.
\end{align*}
Similarly,
\begin{align*}
   \sum_{\lambda\in {\rm spec}(\pi(w))} \sum_{\psi_\alpha\in \widetilde{\cal F}_i^{\pi,\lambda}} |\langle f_{\lambda}, \psi_\alpha \rangle|^2&\leq  \sum_{\lambda\in {\rm spec}(\pi(w))} B_{i}^{\pi,\lambda} \|f_\lambda\|_2^2 \leq \left(\max_{\lambda\in{\rm spec}(\pi(w))}B_{i}^{\pi,\lambda}\right) \|f\|_2^2.
\end{align*}
So, ${\mathcal G}_{\pi,i}$ is a frame for $\E_{\pi,i}$ with lower and upper frame bounds as claimed.
This proves $(i)$.

To prove $(ii)$, consider arbitrary $f\in \ell^2(\bbG)$. Then by Theorem~\ref{thm:frob-schur}, we know that $f$ can be decomposed into an orthogonal direct sum $f=\sum_{\pi\in \widehat{\bbG}, 1\leq i\leq d_\pi}f_{\pi,i}$ with $f_{\pi,i}\in \E_{\pi,i}$.
Since ${\mathcal G}_{\pi,i}$ is a frame for $\E_{\pi,i}$, then we have that 
\begin{equation}\label{eq:large-frame}
   A_{\pi,i}\|f_{\pi,i}\|_2^2\leq \sum_{\phi\in {\mathcal G}_{\pi,i}} |\langle f_{\pi,i},\phi\rangle|^2\leq B_{\pi,i}\|f_{\pi,i}\|_2^2. 
\end{equation}
Since Equation~\eqref{eq:large-frame} holds for all $\pi\in \widehat{\bbG}$ and $1\leq i\leq d_{\pi}$, then summing over $\pi$ and $i$ gives,
\begin{align*}
    \sum_{\pi\in \widehat{\bbG}}\sum_{1\leq i\leq d_{\pi}} \sum_{\phi\in {\mathcal G}_{\pi,i}} |\langle f_{\pi,i}, \phi \rangle|^2 \geq  \sum_{\pi\in \widehat{\bbG}}\sum_{1\leq i\leq d_{\pi}} A_{\pi,i} \|f_{\pi,\lambda}\|_2^2\geq \left(\min_{\pi\in \widehat{\bbG}, 1\leq i\leq d_{\pi}}A_{\pi,i}\right) \|f\|_2^2.
\end{align*}
Similarly,
\begin{align*}
   \sum_{\pi\in \widehat{\bbG}} \sum_{1\leq i\leq d_{\pi}}\sum_{\phi\in {\mathcal G}_{\pi,i}} |\langle f_{\pi,i}, \phi \rangle|^2&\leq  \sum_{\pi\in \widehat{\bbG}}\sum_{1\leq i\leq d_{\pi}} B_{\pi,i}\|f_{\pi,i}\|_2^2\leq\left(\max_{\pi\in \widehat{\bbG},1\leq i\leq d_{\pi}}B_{\pi,i}\right) \|f\|_2^2.
\end{align*}
So, ${\mathcal G}$ is a frame for $\ell^2(\bbG)$ with upper frame bound $\max\{B_{\pi,i}: \pi\in \widehat{\bbG}, 1\leq i\leq d_{\pi}\}$ and lower frame bound $\min\{A_{\pi,i}: \pi\in \widehat{\bbG}, 1\leq i\leq d_{\pi}\}$. That ${\mathcal G}$ is a  Frobenius--Schur frame as well as a $w$-Cayley frame follows directly from Proposition~\ref{prop:Z-pi-decomp}, Remark~\ref{rem:theta} and the definition of ${\mathcal G}$.

Finally, part $(iii)$ follows from the fact that a $w$-Cayley frame can be naturally partitioned into frames for $Z_{\pi,i,\lambda}$. Applying the map $\Theta^{-1}_{\pi,i,\lambda}$ to those frames finishes the proof.
\end{proof}
In the following example, we construct a Cayley frame for $\ell^2(\bbS_4)$ which is not a basis. 
\begin{example}\label{exm:example-non-permutahedron}
Let $G$ be the Cayley graph of $\bbS_4$ with generating set $S=\{(12),(23),(34),(12)(34)\}$. 
  \begin{center}
  \begin{tikzpicture}[scale=0.22]
    \begin{scope}[scale=1, every node/.style={circle,thick,draw}]
        \node (A) at (0,10) [label=above:id]{};
        \node (B) at (2.6,9.7) [label=above:(12)]{};
        \node (C) at (5,8.7) [label=87:(23)]{};
        \node (D) at (7.1,7.1) [label=20:(34)]{};
        \node (E) at (8.7,5) [label=right:(13)]{};
        \node (F) at (9.7,2.6) [label=right:(14)]{};
        \node (G) at (10,0)  [label=right:(24)]{};
        \node (H) at (9.7,-2.6) [label=right:(12)(34)]{};
        \node (I) at (8.7,-5)  [label=355:(13)(24)]{};
        \node (J) at (7.1,-7.1) [label=355:(14)(23)]{};
        \node (K) at (5,-8.7) [label=275:(123)]{};
        \node (L) at (2.6,-9.7) [label=275:(132)]{};
        \node (M) at (0,-10) [label=below:(124)]{};
        \node (N) at (-2.6,-9.7) [label=below:(142)]{};
        \node (O) at (-5,-8.7) [label=265:(134)]{};
        \node (P) at (-7.1,-7.1) [label=255:(143)]{};
        \node (Q) at (-8.7,-5) [label=250:(234)]{};
        \node (R) at (-9.7,-2.6) [label=185:(243)]{};
        \node (S) at (-10,0) [label=left:(1234)]{};
        \node (T) at (-9.7,2.6) [label=170:(1432)]{};
        \node (U) at (-8.7,5) [label=165:(1423)]{};
        \node (V) at (-7.1,7.1) [label=165:(1342)]{};
        \node (W) at (-5,8.7) [label=160:(1324)]{};
        \node (X) at (-2.6,9.7)  [label=95:(1243)]{};
    \end{scope}
    \begin{scope}[>={Stealth[black]},
              every node/.style={fill=white,circle},
              every edge/.style={draw=red,very thick}]
        \draw[thick, teal] (A)--(B);
        \draw[thick, blue] (A)--(H);
        \draw[thick,  purple] (A)--(C);
        \draw[thick, orange] (A)--(D);
        \draw[thick, blue] (B)--(D);
        \draw[thick, orange] (B)--(H);
        \draw[thick, purple] (B)--(K);
        \draw[thick, teal] (C)--(L);
        \draw[thick, orange] (C)--(Q);
        \draw[thick, blue] (C)--(V);
        \draw[thick, teal] (D)--(H);
        \draw[thick, purple] (D)--(R);
        \draw[thick, teal] (E)--(K);
        \draw[thick, purple] (E)--(L);
        \draw[thick, orange] (E)--(O);
        \draw[thick, blue] (E)--(S);
        \draw[thick, purple] (F)--(J);
        \draw[thick, teal] (F)--(M);
        \draw[thick,orange] (F)--(P);
        \draw[thick,blue] (F)--(X);
        \draw[thick, teal] (G)--(N);
        \draw[thick, purple] (G)--(Q);
        \draw[thick,orange] (G)--(R);
        \draw[thick,blue] (G)--(T);
        \draw[thick, purple] (H)--(X);
        \draw[thick,blue] (I)--(J);
        \draw[thick, teal] (I)--(U);
        \draw[thick, purple] (I)--(V);
        \draw[thick,orange] (I)--(W);
        \draw[thick,orange] (J)--(U);
        \draw[thick, teal] (J)--(W);
        \draw[thick,blue] (K)--(O);
        \draw[thick,orange] (K)--(S);
        \draw[thick,blue] (L)--(Q);
        \draw[thick,orange] (L)--(V);
        \draw[thick,blue] (M)--(P);
        \draw[thick, purple] (M)--(S);
        \draw[thick,orange] (M)--(X);
        \draw[thick,blue] (N)--(R);
        \draw[thick,orange] (N)--(T);
        \draw[thick, purple] (N)--(U);
        \draw[thick, teal] (O)--(S);
        \draw[thick, purple] (O)--(W);
        \draw[thick, purple] (P)--(T);
        \draw[thick, teal] (P)--(X);
        \draw[thick, teal] (Q)--(V);
        \draw[thick, teal] (R)--(T);
        \draw[thick,blue] (U)--(W);
    \end{scope}
\end{tikzpicture}
  \captionof{figure}{The Cayley graph of $\bbS_4$ with generating set $S=\{(12), (23), (34), (12)(34)\}.$}
   \label{fig:P_4-weighted}
\end{center}

We use the notation of \ref{Appendix:rep-S4} for the representations of $\bbS_4$ and its coefficient functions.
For each $\phi$ in the set $\{\iota, \tau, \sigma, \pi'\}$, the matrix $\phi(S)$ has only simple eigenvalues, so we have one-dimensional eigenspaces $E_{\lambda}(\phi(S))$, for all $\lambda\in {\rm spec(\phi(S))}$ and $\phi\in \{\iota, \tau, \sigma, \pi'\}$. The matrix $\pi(S)$ has eigenvalue 0 of multiplicity 1 with associated one-dimensional eigenspace $E_{0}(\pi(S))= \left[-\frac{\sqrt{2}}{2},  0,  \frac{\sqrt{2}}{2} \right]^{\t}\bbC$. Also, $\pi(S)$ has eigenvalue $-2$ with multiplicity 2. Fix the following eigenbasis for $E_{-2}(\pi(S))$: 
$$\left\{v_1=\left[ 0, 1, 0\right]^{\t}, v_2=\left[\frac{\sqrt{2}}{2},  0, \frac{\sqrt{2}}{2}\right]^{\t}\right\}.$$ 
We construct the Mercedes-Benz frame for the 2-dimensional eigenspace $E_{-2}(\pi(S))=\text{span}\{v_1,v_2\}$ as follows:
$\left\{v_2, \frac{\sqrt{3}}{2}v_1-\frac{1}{2}v_2, -\frac{\sqrt{3}}{2}v_1-\frac{1}{2}v_2\right\}=
    \left\{
    \begin{bmatrix} \frac{\sqrt{2}}{2}, 0, \frac{\sqrt{2}}{2}\end{bmatrix}^{\t}, 
    \begin{bmatrix} -\frac{\sqrt{2}}{4}, \frac{\sqrt{3}}{2}, -\frac{\sqrt{2}}{4} \end{bmatrix}^{\t}, \begin{bmatrix} -\frac{\sqrt{2}}{4}, -\frac{\sqrt{3}}{2}, -\frac{\sqrt{2}}{4}\end{bmatrix}^{\t}
    \right\}.$
By Theorem~\ref{thm: frame}, the following sets are frames for ${\mathcal E}_{\phi,i}$, where $\phi\in \widehat{\bbS_4}$, $1\leq i\leq d_{\phi}$. 
\begin{itemize}
    \item ${\mathcal G}_{\iota,1}=\left\{\widetilde{\iota_{1,1}}\right\}$
    \item $\mathcal{G}_{\tau, 1}=\left\{\widetilde{\tau_{1,1}}\right\}$
    \item $\mathcal{G}_{\sigma,i}=\left\{\frac{\frac{-\sqrt{3}}{2}+\frac{1}{2}i}{\sqrt{\frac{3}{2}-\frac{\sqrt{3}}{2}i}}\widetilde{\sigma_{1,i}}+\frac{1}{\sqrt{\frac{3}{2}-\frac{\sqrt{3}}{2}}i}\widetilde{\sigma_{2,i}}, \frac{\frac{\sqrt{3}}{2}-\frac{1}{2}i}{\sqrt{\frac{3}{2}-\frac{\sqrt{3}}{2}i}}\widetilde{\sigma_{1,i}}+\frac{1}{\sqrt{\frac{3}{2}-\frac{\sqrt{3}}{2}i}}\widetilde{\sigma_{2,i}}\right\}$, for $i=1,2$
    \item $\mathcal{G}_{\pi',i}=\left\{\widetilde{\pi'_{2,i}},\frac{2-\sqrt{5}}{\sqrt{10-4\sqrt{5}}}\widetilde{\pi'_{1,i}}+\frac{1}{\sqrt{10-4\sqrt{5}}}\widetilde{\pi'_{3,i}}, \frac{\sqrt{5}+2}{\sqrt{10+4\sqrt{5}}}\widetilde{\pi'_{1,i}}+\frac{1}{\sqrt{10+4\sqrt{5}}}\widetilde{\pi'_{3,i}}\right\}$, for $i=1,2,3$
    \item $\mathcal{G}_{\pi,i}=\left\{-\frac{\sqrt{2}}{2}\widetilde{\pi_{1,i}}+\frac{\sqrt{2}}{2}\widetilde{\pi_{3,i}}, 
    \frac{\sqrt{2}}{2}\widetilde{\pi_{1,i}}+\frac{\sqrt{2}}{2}\widetilde{\pi_{3,i}},-\frac{\sqrt{2}}{4}\widetilde{\pi_{1,i}}+\frac{\sqrt{3}}{2}\widetilde{\pi_{2,i}}-\frac{\sqrt{2}}{4}\widetilde{\pi_{3,i}}, -\frac{\sqrt{2}}{4}\widetilde{\pi_{1,i}}-\frac{\sqrt{3}}{2}\widetilde{\pi_{2,i}}-\frac{\sqrt{2}}{4}\widetilde{\pi_{3,i}}\right\}$, for $i=1,2,3$.
\end{itemize}
Moreover $\bigcup_{\phi\in\widehat{\bbS_4}, 1\leq i\leq d_{\phi}} {\mathcal G}_{\phi,i}$ is a \red{${\mathbf 1}_S$-Cayley} frame for $\ell^2(\bbS_4)$.
\end{example}

\subsection{Cayley frames with special properties}\label{subsec:properties}
The frames presented in Theorem \ref{thm: frame} are constructed in such a way that allow us to pass properties from frames for the smaller spaces to frames for the larger space. Recall that \emph{tight frames} are frames for which the upper and lower frame bounds are equal, and \emph{Parseval frames} are tight frames in which the upper and lower frame bounds are both equal to 1. A frame $\mathcal F=\{\phi_i\}_{i\in I}$ is a \textit{unit norm frame} if $\|\phi_i\|_{2}=1$ for all $i\in I$.

The following corollary is an immediate consequence of Theorem \ref{thm: frame} $(i)$ and $(ii)$. 
\begin{corollary}\label{cor:frame-bounds}
 With notation as in Theorem~\ref{thm: frame}, suppose the frames ${\cal F}^{\pi,\lambda}_i$  are Parseval  (resp.~tight, resp.~unit norm) for all $\pi$, 
$i$, and $\lambda$. 
Then ${\mathcal G}_{\pi,i}$ is a Parseval (resp.~tight, resp.~unit norm) frame for $\E_{\pi, i}$, and ${\mathcal G}$ is a Parseval (resp.~tight, resp.~unit norm) frame for $\ell^2(\bbG)$.
\end{corollary}

To introduce the next property, we fix a basis $\mathcal B$ for $\ell^2(\bbG)$. 
A vector $x\in \ell^2(\bbG)$ is called \emph{$K$-sparse} (with respect to $\mathcal B$) if it is a linear combination of at most $K$ basis elements.  We call $\mathcal{B}$  the \textit{sparsity basis} for $\ell^2(\bbG)$.
A unit norm frame $\mathcal F$ for $\ell^2(\bbG)$ is said to have the \emph{restricted isometry property} of order $K$ with parameter $\delta_K \in (0,1)$ if
$$(1-\delta_K) \|x\|_2^2 \leq \|\mathcal Fx\|_2^2 \leq (1+\delta_K) \|x\|_2^2,$$
for all $K$-sparse $x\in \ell^2(\bbG)$ \cite{Foucart:2013:CompressiveSensing}.
The restricted isometry property was introduced in \cite{Candes:2005:DecodingLinearProgram} and \cite{Candes:2008:RIPImplicationsCompressedSensing}, and is of particular importance in compressive sensing. Given a vector $y\in \mathbb C^m$ of observed data, an unknown sparse signal $x\in \mathbb C^N$, and a measurement matrix $A\in \mathbb C^{m\times N}$, the goal of compressive sensing is to determine what matrices $A$ allow for sparse reconstruction of the signal $x$ in the underdetermined $(m<N)$ system $y=Ax$, and to develop efficient reconstruction methods. The restricted isometry property is key to obtaining an optimal lower bound on $m$, the number of measurements, in terms of $N$ and the sparsity of $x$. Also, frames with small $\delta_K$ for sufficiently large $K$ are more desirable, as they can serve as measurement matrices that allow for successful sparse reconstruction \cite{Foucart:2013:CompressiveSensing}. The frame construction in Theorem \ref{thm: frame} allows us to control the restricted isometry property constants $\delta_K$.

Let $\mathcal{B}=\bigcup_{\pi\in\widehat{\bbG}}\bigcup_{k:1,\ldots,d_\pi}\left\{\sqrt{\frac{d_\pi}{|\bbG|}}\pi_{j,k}: 1\leq j\leq d_{\pi}\right\}$, which is the preferred Fourier basis for $\ell^2(\bbG)$, and let $\mathcal{B}_{\pi,i}=\left\{\sqrt{\frac{d_\pi}{|\bbG|}}\pi_{j,i}: 1\leq j\leq d_{\pi}\right\}$, which we think of as the preferred Fourier basis for $\mathcal{E}_{\pi,i}$. 

\begin{corollary}
With notation as in Theorem~\ref{thm: frame}, suppose each frame ${\cal G}_{\pi,i}$  has the restricted isometry property of order $K$ with parameter $\delta_K^{\pi,i}$, with respect to the basis $\mathcal{B}_{\pi,i}$. Then $
\mathcal G$ has the restricted isometry property of order $K$ with parameter $\delta:=\max_{\pi,i} \delta_K^{\pi,i}$, with respect to the basis $\mathcal{B}$. 
\end{corollary}

\begin{proof}
Suppose $x\in \ell^2(\bbG)$ is $K$-sparse with respect to $\mathcal{B}$, and define $P_{\pi,i}:\ell^2(\mathbb G)\rightarrow \ell^2(\mathbb G)$ to be the orthogonal projection onto the space $\mathcal{E}_{\pi,i}$. Note that for each $\pi\in\mathbb{\widehat{G}}$ and each $1\leq i\leq d_{\pi}$, the projection $P_{\pi,i}$ preserves sparsity. That is, $P_{\pi,i}x$ is $K$-sparse with respect to the basis $\mathcal{B}_{\pi,i}$ for $\mathcal{E}_{\pi,i}$. 
So, by the assumption, we have for each $\pi\in\widehat{\bbG}$ and $1\leq i\leq d_{\pi}$ that 
\begin{align*}
(1-\delta_{K}^{\pi,i})\|P_{\pi,i}x\|_2^2\leq \sum_{\phi\in \mathcal{G}_{\pi,i}} |\langle P_{\pi,i}x, \phi\rangle|^2 \leq (1+\delta_{K}^{\pi,i}) \|P_{\pi,i}x\|_2^2.
\end{align*}
Now define $\delta=\max_{\pi,i,\lambda} \delta_{K}^{\pi,i,\lambda}$, and observe that
\begin{align*}
    \|\mathcal{G}x\|_2^2=\sum_{\pi\in \widehat{\mathbb G}} \sum_{1\leq i \leq d_{\pi}}  \sum_{\phi\in \mathcal{G}_{\pi,i}} |\langle P_{\pi,i}x, \phi\rangle|^2\leq \sum_{\pi\in \widehat{\mathbb G}} \sum_{1\leq i \leq d_{\pi}}  (1+\delta)\|P_{\pi,i}x\|_2^2= (1+\delta)\|x\|_2^2.
\end{align*}
The lower bound follows by a similar argument. Therefore, $\mathcal{G}$ has the restricted isometry property of order $K$ with parameter $\delta$, with respect to the basis $\mathcal{B}$.
\end{proof}

\subsection{Comparison with previous work: frames for the permutahedron}\label{subsec:relation-to-paper}
In \cite[Equation (11)]{paper}, suitable frames for the permutahedron were proposed. The same frames can be produced as an immediate application of Theorem \ref{thm: frame} to the orthonormal eigenbasis of Proposition~\ref{prop:eigenvector}. Namely, Proposition~\ref{prop:eigenvector}, when applied to the case of the permutahedron ${\mathbb P}_n$, provides us with a second method for obtaining the decompositions in \cite[Proposition 1]{paper}. 
A precise definition of the permutahedron is given in Subsection~\ref{subsec:eigen-decomp of A}.
We note that, unlike the methods in \cite{paper}, our approach does not make use of equitable partitions and Schreier graphs. Instead, we need a full description (in matrix form) of all irreducible representations of $\bbS_n$ to obtain a concrete description of the eigenvectors of the adjacency matrix of ${\mathbb P}_n$ (see Example~\ref{example:P3} and Example~\ref{example:P4}).
Here, we focus on the relation between our notations and results and those of \cite{paper}.

Let $\pi$ be an irreducible representation of $\bbS_n$ associated with a partition $\gamma$ of $[n]$. Then the space $W_\gamma$ given in \cite[Equation (2)]{paper} is simply the subspace of $\ell^2(\bbS_n)$ containing all coefficient functions associated with $\pi$, namely, 
$$W_\gamma=\oplus_{i=1}^{d_\pi}{\cal E}_{\pi,i}.$$
Fix an eigenvalue $\lambda$ of the adjacency matrix $A_{{\mathbb P}_n}$. Proposition~\ref{prop:eigenvector} $(ii)$ states that every $\lambda$-eigenvector $\phi\in \bbC^{|\bbS_n|}$ of $A_{{\mathbb P}_n}$ can be written as
$$\phi=B(\oplus_{\sigma\in \widehat{\bbS_n}}\oplus_{i=1}^{d_\sigma} X_{\sigma,i}),$$
where $B$ is the unitary matrix of normalized coefficient functions given in \eqref{eq:matrix-B}, and every $X_{\sigma,i}\in \bbC^{d_\sigma}$ is either 0 or a $\lambda$-eigenvector for $\sigma(S)=\sum_{s\in S}\sigma(s)$. Thus $\phi\in W_\gamma$ precisely when $X_{\sigma,i}=0$ whenever $\sigma\neq \pi$.
So, a $\lambda$-eigenvector $\phi$ of $A_{{\mathbb P}_n}$ belongs to $W_\gamma$ if and only if  it can be written as 
$$\Theta_{\pi,1}(X_1)+\Theta_{\pi,2}(X_2)+\ldots +\Theta_{\pi,d_\pi}(X_{d_\pi}),$$ 
for $X_1,\ldots, X_{d_\pi}\in E_\lambda(\pi(S))$  where at least one of the $X_i$'s is nonzero. Finally, by Remark~\ref{rem:theta} $(ii)$,
we have $\Theta_{\pi,j}(X_j)\in Z_{\pi,j,\lambda}$ for each $1\leq j\leq d_\pi$.
The space $U_{\lambda}$ given in \cite[Equation (6)]{paper} is defined as the $\lambda$-eigenspace of the adjacency matrix $A_{\mathbb P_n}$.  
The above discussion shows that the space $Z_{\gamma,\lambda}=W_{\gamma}\cap U_{\lambda}$ in \cite[Proposition 1]{paper} is exactly the space $Z_{\pi,\lambda}=\oplus_{i=1}^{d_\pi} Z_{\pi,i,\lambda}$, where $Z_{\pi,i,\lambda}$ is given in Definition \ref{def:frame-compatible}$(ii)$.

\section{Future work}
\red{In this paper, we provide a complete description of the eigen-decomposition of the adjacency matrix of a weighted Cayley graph in terms of the irreducible representations of its underlying group (Proposition~\ref{prop:eigenvector}). 
Leveraging this eigen-decomposition, we characterize all weighted Cayley frames (with respect to a given weight function) of $\ell^2(\bbG)$, and offer a concrete method for constructing such frames (Theorem~\ref{thm: frame}).
In a weighted Cayley frame, each atom is associated exclusively with the coefficient space of a single irreducible representation of the underlying group.
So, essentially, we obtain bases/frames that are efficient in the Fourier domain, analogous to the role played by the classical Fourier basis in classical signal processing.

We believe that this choice of eigenbasis/frame, informed by the representation theory of the underlying group, leads to more efficient tools for signal processing. As a follow-up to this work, we plan to carry out numerical experiments or simulations to explore the following scenarios.
\begin{itemize}
\item[(i)] It is reasonable to expect that our representation-theoretic approach should be especially beneficial when a repeated eigenvalue is associated with two different irreducible representations. With terminology of Proposition~\ref{prop:eigenvector}, this means that the value $\lambda$ is an eigenvalue for both $\pi(w)$ and $\sigma(w)$, with $\pi,\sigma\in\widehat{\bbG}$. We plan to run numerical experiments to compare the efficiency of our construction (where the eigenvalue $\lambda$ associated with each representation $\pi$ and $\sigma$ is treated separately) with general GSP constructions in the literature.

\item[(ii)] In the context of ranked data analysis, the generating set captures a notion of ``closeness'' between rankings. It is reasonable to assume that in different voting contexts, it might be useful to choose a generating set different from that of the permutahedron. In particular, we are interested in analyzing scenarios where the top positions in the ranking are considered to be much more important than the bottom positions. This could be the case where $n$ candidates are ranked for a single job opening. In that case, two rankings that differ only in the positions $n-1$ and $n$ are intuitively considered to be more similar than rankings in which positions 1 and 2 are switched. We plan to pursue this direction in future work. 
We have plans to analyze various data sets in the context of ranked data to see
how varying the generating set affects the interpretation of the data and the efficacy of the developed GSP tools.
\end{itemize}
}

\section*{Acknowledgements}
The authors thank the Department of Mathematical Sciences at University of Delaware for their continued support throughout the process of this research. 
This project was initiated during a summer research program funded by University of Delaware
UNIDEL summer fellowship. The first author thanks University of Delaware for funding her UNIDEL Ph.D.~project in summer 2021. The second author acknowledges support from National Science Foundation Grant DMS-1902301 and DMS-2408008 during the preparation of this article.
\red{We sincerely thank the anonymous reviewers for critically reading the manuscript and suggesting substantial improvements. }


\appendix
\section{Proof of Proposition \ref{prop:eigenvector}} \label{Appendix:proofs}

\begin{proof}
%
First note that for every $\pi\in \widehat{\bbG}$, the $d_\pi\times d_\pi$ matrix $\pi(w)$ is Hermitian, because 
$$\pi(w)^*=\left(\sum_{x\in \bbG}w(x)\pi(x)\right)^*=\sum_{x\in \bbG}\overline{w(x)}\pi(x^{-1})=\sum_{x\in \bbG}w(x)\pi(x)=\pi(w).$$
Hence $\pi(w)$ is unitarily diagonalizable. That is,  $\bbC^{d_\pi}$ admits an orthonormal basis consisting of eigenvectors of $\pi(w)$.

The adjacency matrix $A_G$ can be written in terms of the right regular representation as follows:
\begin{equation}\label{eq1:proof}
A_G=\sum_{x\in \bbG} w(x)\rho(x).
\end{equation}
By the Frobenius--Schur theorem, there exists a unitary matrix $B$ such that $\rho(g)$ can be block diagonalized. Namely,
\begin{equation}\label{eq2:proof}
B^{-1}\rho(g)B=\bigoplus_{\pi\in \widehat{\bbG}}d_\pi\cdot\pi(g)\ \ \forall g\in\bbG,
\end{equation}
where $d_\pi$ denotes the dimension of $\pi$, and $d_\pi\cdot\pi(g)$ denotes the direct sum of $d_\pi$ copies of $\pi(g)$.
Putting \eqref{eq1:proof} and \eqref{eq2:proof} together, we get a block diagonalization of $A_G$ as follows:
\begin{equation}\label{eq:thm-directsum}
B^{-1}A_GB= B^{-1}\left(\sum_{x\in \bbG} w(x)\rho(x)\right)B=\bigoplus_{\pi\in \widehat{\bbG}}d_\pi\cdot\left(\sum_{x\in \bbG}w(x)\pi(x)\right)=\bigoplus_{\pi\in \widehat{\bbG}}d_\pi\cdot\pi(w).
\end{equation}
Let $N$ denote the size of $\bbG$, and consider $\phi\in \bbC^N$.
From \eqref{eq:thm-directsum}, the vector $\phi$ is a $\lambda$-eigenvector of $A_G$ precisely when $B^{-1}\phi$ is a $\lambda$-eigenvector of 
$\bigoplus_{\pi\in \widehat{\bbG}}d_\pi\cdot\pi(w)$, because 
\begin{equation}\label{eq:thm-eigen-iff}
\lambda(B^{-1}\phi)=B^{-1}(\lambda\phi)= B^{-1}A_G\phi=\left(\bigoplus_{\pi\in \widehat{\bbG}}d_\pi\cdot \pi(w)\right)B^{-1}\phi.
\end{equation}
This finishes the proof of $(i)$, because the spectrum of the block diagonal matrix $\bigoplus_{\pi\in \widehat{\bbG}}d_\pi\cdot \pi(w)$ is the collection of eigenvalues of its blocks. 

To prove $(ii)$, consider the block partition of the vector $B^{-1}\phi\in \bbC^N$ that is compatible with the decomposition $\bigoplus_{\pi\in \widehat{\bbG}}d_\pi\cdot\pi(w)$, i.e.~$B^{-1}\phi=\oplus_{\pi\in \widehat{\bbG}}\oplus_{i=1}^{d_\pi} X_{\pi,i},$ with $X_{\pi,i}\in \bbC^{d_\pi}$.
So by \eqref{eq:thm-eigen-iff}, the vector $\phi$ is a $\lambda$-eigenvector of $A_G$ if and only if $\lambda X_{\pi,i}=\pi(w)X_{\pi,i}$ for every $\pi\in\widehat{\bbG}$ and every appropriate $i$. That is, $\phi$ is a $\lambda$-eigenvector of $A_G$ if and only if the following two conditions hold:
\begin{itemize}
    \item[(a)] for every $\pi\in\widehat{\bbG}$ and $1\leq i\leq d_\pi$, either $X_{\pi,i}=0$ or $X_{\pi,i}$ is a $\lambda$-eigenvector of $\pi(w)$;
    \item[(b)] there exists some $\pi\in\widehat{\bbG}$ and $1\leq i\leq d_\pi$ such that $X_{\pi,i}\neq 0$.
\end{itemize}
By part $(i)$, every matrix $\pi(w)$ is diagonalizable. For every $\pi\in\widehat{\bbG}$, let $\{V_1^\pi,\ldots, V_{d_\pi}^\pi\}\subseteq\bbC^{d_\pi}$ be an orthonormal eigenbasis for $\pi(w)$. 
Every fixed eigenvector $V_j^\pi$ of $\pi(w)$ can be naturally embedded in $\oplus_{\pi\in \widehat{\bbG}}\oplus_{i=1}^{d_\pi} \bbC^{d_\pi}$ in $d_\pi$-many ways; we denote these embedded versions of $V^{\pi}_j$ by $V_j^{\pi,i}\in \bbC^N$, where $i=1,\ldots, d_\pi$. 
More precisely, for $i=1,\ldots, d_\pi$, we define
$$V_j^{\pi,i}=\oplus_{\sigma\in \widehat{\bbG}}\oplus_{k=1}^{d_\sigma} Y_{\sigma,k}, \mbox{ where } 
Y_{\sigma,k}=\left\{
\begin{array}{cc}
V_j^{\pi}   & \mbox{ if } \sigma=\pi \mbox{ and } k=i \\
   0  & \mbox{ otherwise}
\end{array}\right..$$
From the definition of the vectors $V_j^{\pi,i}$, it is easy to see that the set $\bigcup_{\pi\in\widehat{\bbG}}\left\{V_j^{\pi,i}:\, 1\leq i,j\leq d_{\pi}\right\}$ is an orthonormal set in $\bbC^N$. Moreover, this set has $\sum_{\pi\in\widehat{\bbG}} d_\pi^2=N$ elements; so it forms an orthonormal basis for $\bbC^N$. Finally, since $B$ is a unitary matrix, we obtain the following orthonormal basis of $\bbC^N$:
$${\cal B}:= \bigcup_{\pi\in\widehat{\bbG}}\left\{BV_j^{\pi,i}:\ 1\leq i,j\leq d_{\pi}\right\}.$$ 
By properties (a) and (b) above, the collection ${\cal B}$ is indeed an orthonormal basis consisting of eigenvectors of $A_G$. This finishes the proof of $(ii)$.

Given the explicit form of the matrix $B$ (see \eqref{eq:matrix-B}), we observe that 
$$BV_j^{\pi,i}=\sqrt{\frac{d_\pi}{|\bbG|}}\Big[\pi_{1,i}|\pi_{2,i}|\ldots|\pi_{d_\pi,i}\Big]V_j^{\pi}=\sqrt{\frac{d_\pi}{|\bbG|}}\sum_{k=1}^{d_\pi}x_k\pi_{k,i},$$
where $V_j^{\pi}=\left[x_1, \cdots, x_{d_{\pi}}\right]^{\t}$.
\end{proof}

\section{Representations of $\bbS_n$}\label{Appendix:rep-Sn}
The symmetric group $\bbS_n$ is the group of all permutations on $n$ elements with composition of permutations as the group operation. As we use the irreducible representations of $\mathbb S_n$ in our approach, we recall that a finite group's irreducible representations are in one-to-one correspondence with its conjugacy classes \cite[Proposition 2.30]{FultonHarris}. For $\mathbb{S}_n$, the conjugacy classes are determined by the \emph{cycle type} of the permutations, where the cycle type describes the number of cycles and their lengths in the unique cycle decomposition of a permutation. 
%
Therefore, the conjugacy classes of $\mathbb S_n$ are in bijective correspondence with the partitions $\lambda \vdash n$ of $n$.

A useful accounting technique for partitions is to represent them as a collection of rows and columns of boxes, called a Young diagram. A \emph{Young diagram} with shape $\lambda = (\lambda_1,\ldots,\lambda_k)$, where $\{\lambda_i\}_{i=1}^k$ is in non-increasing order, has $\lambda_i$ boxes in its $i$th row, for every $1\leq i\leq k$. Thus, the number of Young diagrams with $n$ blocks is exactly the number of partitions of $n$. A Young diagram can be extended to a \emph{Young tableau}, which is a Young diagram on $n$ blocks where each block is uniquely labeled from the set $\{1,2,\dots, n-1,n\}$. A Young tableau is said to be in \emph{standard form} if the labels in each row increase from left to right and the labels in each column increase from top to bottom. While Young diagrams correspond with irreducible representations of $\bbS_n$, the number of Young tableaux for a given Young diagram (or, partition $\lambda \vdash n$) is the dimension of the corresponding irreducible representation of $\bbS_n$.

%

%

\subsection{Representations of $\bbS_3$ in matrix form and their matrix coefficients}\label{Appendix:rep-S3}
To write functions on $\bbS_3$ as vectors in $\bbC^6$, we order elements of $\bbS_3$ as follows:  $\id,(12),(23),(13),(123),(132)$. 
Below we have provided all irreducible representations of $\bbS_3$, expressed in matrix form, as well as the coefficient functions of each representation.
\begin{itemize}
\item[(i)] The {\bf trivial representation} of $\bbS_3$, denoted by $\iota$, is the 1-dimensional representation that maps every element of $\bbS_3$
to 1.  This representation is associated with the partition $3,0,0$ of $n=3$. The unique coefficient function of $\iota$ is given by $\iota_{1,1}=\left[1,1,1,1,1,1\right]^{\t}.$

\item[(ii)] The {\bf alternating representation} of $\bbS_3$, denoted by $\tau$, is the 1-dimensional representation that maps $\sigma\in \bbS_3$ to the sign of the permutation. This representation is associated with the partition $1,1,1$ of $n=3$.
The unique coefficient function of $\tau$ is given by $\tau_{1,1}=\left[1,-1,-1,-1,1,1\right]^{\t}.$
\item[(iii)]  The {\bf standard representation} of $\bbS_3$ is the 2-dimensional irreducible representation $\pi:\bbS_3\rightarrow U_2(\C)$ defined as
follows:
$$
    \pi(\id)=\begin{bmatrix} 1 & 0\\ 0 & 1\end{bmatrix}, \ 
    \pi((12))=\begin{bmatrix} -\frac{1}{2} & \frac{\sqrt{3}}{2}\\ \frac{\sqrt{3}}{2} & \frac{1}{2}\end{bmatrix}, \ 
    \pi((23))=\begin{bmatrix} 1 & 0\\ 0 & -1\end{bmatrix}.
$$
Since $\pi$ is multiplicative, and we have $(13)=(12)(23)(12)$ and $(123)=(12)(23)$, the above matrices are enough to define $\pi$ on $\bbS_3$.
This representation is associated with the partition $2,1,0$ of $n=3$.
The normalized coefficient functions of $\pi$ are
$$
    \sqrt{\frac13}\pi_{1,1}=\begin{bmatrix}\frac{1}{\sqrt{3}}\\ -\frac{1}{2\sqrt{3}}\\ \frac{1}{\sqrt{3}}\\ -\frac{1}{2\sqrt{3}}\\ -\frac{1}{2\sqrt{3}}\\-\frac{1}{2\sqrt{3}}\end{bmatrix}, \ 
    \sqrt{\frac13}\pi_{2,1}=\begin{bmatrix} 0\\ \frac{1}{2}\\ 0\\ -\frac{1}{2}\\ -\frac{1}{2}\\ \frac{1}{2}\end{bmatrix}, \
    \sqrt{\frac13}\pi_{1,2}=\begin{bmatrix} 0\\ \frac{1}{2}\\ 0\\ -\frac{1}{2}\\ \frac{1}{2}\\ -\frac{1}{2}\end{bmatrix}, \ 
    \sqrt{\frac13}\pi_{2,2}=\begin{bmatrix} \frac{1}{\sqrt{3}}\\ \frac{1}{2\sqrt{3}}\\ -\frac{1}{\sqrt{3}}\\ \frac{1}{2\sqrt{3}}\\ -\frac{1}{2\sqrt{3}}\\ -\frac{1}{2\sqrt{3}}\end{bmatrix}.
$$
\end{itemize}

\subsection{Representations of $\bbS_4$ in matrix form and their matrix coefficients} \label{Appendix:rep-S4}
In order to identify functions on $\bbS_4$ with vectors of size 24, we order elements of $\bbS_4$ as follows:  
$$\id,(12),(23),(34),(13),(14),(24),(12)(34),(13)(24),(14)(23),(123),(132),(124),(142),(134),$$
$$(143),(234),(243),(1234),(1432),(1423),(1342),(1324),(1243).$$
Below, we describe the irreducible unitary representations of $\bbS_4$. Since the elements $(12),(23),$ and $(34)$ generate the group, we first describe the explicit matrix representations with respect to these generating elements. The remaining matrix representations are given in Table \ref{table:unitary-rep}.

\begin{itemize}
\item[(i)] The {\bf trivial representation} of $\bbS_4$, denoted by $\iota$, is the 1-dimensional representation that maps every element of $\bbS_4$
to 1.  This representation is associated with the partition $4,0,0,0$ of $n=4$. 
\item[(ii)] The {\bf alternating representation} of $\bbS_4$, denoted by $\tau$, is the 1-dimensional representation that maps $\sigma\in \bbS_4$ to the sign of the permutation. This representation is associated with the partition $1,1,1,1$ of $n=4$.
\item[(iii)]  The {\bf 2-dimensional representation} of $\bbS_4$ denoted by $\sigma$ is defined as follows:
$$\sigma (\id)=
\begin{bmatrix}
1& 0\\
0& 1
\end{bmatrix}, \ 
\sigma ((12))=
\begin{bmatrix}
0& 1\\
1& 0
\end{bmatrix}, \ 
\sigma ((23))=
\begin{bmatrix}
0& \omega ^2\\
\omega& 0
\end{bmatrix}, \ 
\sigma ((34))=
\begin{bmatrix}
0& 1\\
1& 0
\end{bmatrix},
$$
where $w$ is the third root of unity.
This representation is associated with the partition $2,2,0,0$ of $n=4$. 
\item[(iv)]  The {\bf standard-alternating tensor representation} of $\bbS_4$, denoted by $\pi$, is 3-dimensional and is defined as follows:
$$
\pi (\id)=
\begin{bmatrix}
1& 0& 0\\
0& 1& 0\\
0& 0& 1
\end{bmatrix}, \ 
\pi ((12))=
\begin{bmatrix}
-1& 0& 0\\
0& 0& 1\\
0& 1& 0
\end{bmatrix}, \ 
\pi ((23))=
\begin{bmatrix}
0& 0& -1\\
0& -1& 0\\
-1& 0& 0
\end{bmatrix}, \ 
\pi ((34))=
\begin{bmatrix}
-1& 0& 0\\
0& 0& -1\\
0& -1& 0
\end{bmatrix}.
$$
This representation is associated with the partition $2, 1, 1, 0$ of $n=4$. 
\item[(v)]  The {\bf standard representation} of $\bbS_4$, denoted by $\pi'$, is 3-dimensional and is defined as follows:
$$
\pi' (\id)=
\begin{bmatrix}
1& 0& 0\\
0& 1& 0\\
0& 0& 1
\end{bmatrix}, \ 
\pi' ((12))=
\begin{bmatrix}
1& 0& 0\\
0& 0& -1\\
0& -1& 0
\end{bmatrix}, \ 
\pi' ((23))=
\begin{bmatrix}
0& 0& 1\\
0& 1& 0\\
1& 0& 0
\end{bmatrix}, \ 
\pi' ((34))=
\begin{bmatrix}
1& 0& 0\\
0& 0& 1\\
0& 1& 0
\end{bmatrix}.
$$
This representation is associated with the partition $3, 1, 0, 0$ of $n=4$. 
\end{itemize}

In Table \ref{table:unitary-rep}, we summarize the above unitary irreducible representations of $\bbS_4$, and provide the explicit matrix forms for all group elements. Here, $\omega$ denotes the third root of unity. In Table \ref{table:unitary-rep-matrix-coef} we provide the matrix coefficients of the unitary representations of $\bbS_4$.

%
%

\begin{table}[h]
\centering
{
$\begin{array}{|c|c|c|c|c|c|}
\hline
\mathbf{\bbS_4}& \boldsymbol\iota& \boldsymbol\tau& \boldsymbol\sigma& \boldsymbol\pi& \boldsymbol\pi'\\
\hline
\textbf{id}& 1& 1& \left[\begin{smallmatrix}
1& 0\\ 0& 1 
\end{smallmatrix}\right]& \tiny\left[\begin{smallmatrix}
1& 0& 0\\ 0& 1& 0\\ 0& 0& 1
\end{smallmatrix}\right]& \tiny\left[\begin{smallmatrix}
1& 0& 0\\ 0& 1& 0\\ 0& 0& 1
\end{smallmatrix}\right]\\
\hline
\mathbf{(12)}& 1& -1& \left[\begin{smallmatrix}
0& 1\\ 1& 0 
\end{smallmatrix}\right]& \tiny\left[\begin{smallmatrix}
-1& 0& 0\\ 0& 0& 1\\ 0& 1& 0
\end{smallmatrix}\right]& \tiny\left[\begin{smallmatrix}
1& 0& 0\\ 0& 0& -1\\ 0& -1& 0
\end{smallmatrix}\right]\\
\hline
\mathbf{(23)}& 1& -1& \left[\begin{smallmatrix}
0& \omega ^2\\ \omega& 0
\end{smallmatrix}\right]& \tiny\left[\begin{smallmatrix}
0& 0& -1\\ 0& -1& 0\\ -1& 0& 0 \end{smallmatrix}\right]&
\tiny\left[\begin{smallmatrix}
0& 0& 1\\ 0& 1& 0\\ 1& 0& 0  \end{smallmatrix}\right]\\
\hline
\mathbf{(34)}& 1& -1& \left[\begin{smallmatrix}
0& 1\\ 1& 0
\end{smallmatrix}\right]& \tiny\left[\begin{smallmatrix}
-1& 0& 0\\ 0& 0& -1\\ 0& -1& 0 \end{smallmatrix}\right]&
\tiny\left[\begin{smallmatrix}
1& 0& 0\\ 0& 0& 1\\ 0& 1& 0  \end{smallmatrix}\right]\\
\hline
\mathbf{(13)}& 1& -1& \left[\begin{smallmatrix}
0& \omega\\ \omega ^2& 0
\end{smallmatrix}\right]& \tiny\left[\begin{smallmatrix}
0& 1& 0\\ 1& 0& 0\\ 0& 0& -1 \end{smallmatrix}\right]&
\tiny\left[\begin{smallmatrix}
0& -1& 0\\ -1& 0& 0\\ 0& 0& 1  \end{smallmatrix}\right]\\
\hline
\mathbf{(14)}& 1& -1& \left[\begin{smallmatrix}
0& \omega ^2\\ \omega& 0
\end{smallmatrix}\right]& \tiny\left[\begin{smallmatrix}
0& 0& 1\\ 0& -1& 0\\ 1& 0& 0 \end{smallmatrix}\right]&
\tiny\left[\begin{smallmatrix}
0& 0& -1\\ 0& 1& 0\\ -1& 0& 0 \end{smallmatrix}\right]\\
\hline
\mathbf{(24)}& 1& -1& \left[\begin{smallmatrix}
0& \omega\\ \omega ^2& 0
\end{smallmatrix}\right]& \tiny\left[\begin{smallmatrix}
0& -1& 0\\ -1& 0& 0\\ 0& 0& -1 \end{smallmatrix}\right]&
\tiny\left[\begin{smallmatrix}
0& 1& 0\\ 1& 0& 0\\ 0& 0& 1 \end{smallmatrix}\right]\\
\hline
\mathbf{(12)(34)}& 1& 1& \left[\begin{smallmatrix}
1& 0\\ 0& 1
\end{smallmatrix}\right]& \tiny\left[\begin{smallmatrix}
1& 0& 0\\ 0& -1& 0\\ 0& 0& -1 \end{smallmatrix}\right]&
\tiny\left[\begin{smallmatrix}
1& 0& 0\\ 0& -1& 0\\ 0& 0& -1 \end{smallmatrix}\right]\\
\hline
\mathbf{(13)(24)}& 1& 1& \left[\begin{smallmatrix}
1& 0\\ 0& 1
\end{smallmatrix}\right]& \tiny\left[\begin{smallmatrix}
-1& 0& 0\\ 0& -1& 0\\ 0& 0& 1 \end{smallmatrix}\right]&
\tiny\left[\begin{smallmatrix}
-1& 0& 0\\ 0& -1& 0\\ 0& 0& 1 \end{smallmatrix}\right]\\
\hline
\mathbf{(14)(23)}& 1& 1& \left[\begin{smallmatrix}
1& 0\\ 0& 1
\end{smallmatrix}\right]& \tiny\left[\begin{smallmatrix}
-1& 0& 0\\ 0& 1& 0\\ 0& 0& -1 \end{smallmatrix}\right]&
\tiny\left[\begin{smallmatrix}
-1& 0& 0\\ 0& 1& 0\\ 0& 0& -1 \end{smallmatrix}\right]\\
\hline
\mathbf{(123)}& 1& 1& \left[\begin{smallmatrix}
\omega& 0\\ 0& \omega ^2
\end{smallmatrix}\right]& \tiny\left[\begin{smallmatrix}
0& 0& 1\\ -1& 0& 0\\ 0& -1& 0 \end{smallmatrix}\right]&
\tiny\left[\begin{smallmatrix}
0& 0& 1\\ -1& 0& 0\\ 0& -1& 0 \end{smallmatrix}\right]\\
\hline
\mathbf{(132)}& 1& 1& \left[\begin{smallmatrix}
\omega ^2& 0\\ 0& \omega
\end{smallmatrix}\right]& \tiny\left[\begin{smallmatrix}
0& -1& 0\\ 0& 0& -1\\ 1& 0& 0 \end{smallmatrix}\right]&
\tiny\left[\begin{smallmatrix}
0& -1& 0\\ 0& 0& -1\\ 1& 0& 0 \end{smallmatrix}\right]\\
\hline
\mathbf{(124)}& 1& 1& \left[\begin{smallmatrix}
\omega ^2& 0\\ 0& \omega
\end{smallmatrix}\right]& \tiny\left[\begin{smallmatrix}
0& 1& 0\\ 0& 0& -1\\ -1& 0& 0 \end{smallmatrix}\right]&
\tiny\left[\begin{smallmatrix}
0& 1& 0\\ 0& 0& -1\\ -1& 0& 0 \end{smallmatrix}\right]\\
\hline
\mathbf{(142)}& 1& 1& \left[\begin{smallmatrix}
\omega& 0\\ 0& \omega ^2
\end{smallmatrix}\right]& \tiny\left[\begin{smallmatrix}
0& 0& -1\\ 1& 0& 0\\ 0& -1& 0 \end{smallmatrix}\right]&
\tiny\left[\begin{smallmatrix}
0& 0& -1\\ 1& 0& 0\\ 0& -1& 0 \end{smallmatrix}\right]\\
\hline
\mathbf{(134)}& 1& 1& \left[\begin{smallmatrix}
\omega& 0\\ 0& \omega ^2
\end{smallmatrix}\right]& \tiny\left[\begin{smallmatrix}
0& 0& -1\\ -1& 0& 0\\ 0& 1& 0 \end{smallmatrix}\right]&
\tiny\left[\begin{smallmatrix}
0& 0& -1\\ -1& 0& 0\\ 0& 1& 0 \end{smallmatrix}\right]\\
\hline
\mathbf{(143)}& 1& 1& \left[\begin{smallmatrix}
\omega ^2& 0\\ 0& \omega
\end{smallmatrix}\right]& \tiny\left[\begin{smallmatrix}
0& -1& 0\\ 0& 0& 1\\ -1& 0& 0 \end{smallmatrix}\right]&
\tiny\left[\begin{smallmatrix}
0& -1& 0\\ 0& 0& 1\\ -1& 0& 0 \end{smallmatrix}\right]\\
\hline
\mathbf{(234)}& 1& 1& \left[\begin{smallmatrix}
\omega ^2& 0\\ 0& \omega
\end{smallmatrix}\right]& \tiny\left[\begin{smallmatrix}
0& 1& 0\\ 0& 0& 1\\ 1& 0& 0 \end{smallmatrix}\right]&
\tiny\left[\begin{smallmatrix}
0& 1& 0\\ 0& 0& 1\\ 1& 0& 0 \end{smallmatrix}\right]\\
\hline
\mathbf{(243)}& 1& 1& \left[\begin{smallmatrix}
\omega& 0\\ 0& \omega ^2
\end{smallmatrix}\right]& \tiny\left[\begin{smallmatrix}
0& 0& 1\\ 1& 0& 0\\ 0& 1& 0 \end{smallmatrix}\right]&
\tiny\left[\begin{smallmatrix}
0& 0& 1\\ 1& 0& 0\\ 0& 1& 0 \end{smallmatrix}\right]\\
\hline
\mathbf{(1234)}& 1& -1& \left[\begin{smallmatrix}
0& \omega\\ \omega ^2& 0
\end{smallmatrix}\right]& \tiny\left[\begin{smallmatrix}
0& -1& 0\\ 1& 0& 0\\ 0& 0& 1 \end{smallmatrix}\right]&
\tiny\left[\begin{smallmatrix}
0& 1& 0\\ -1& 0& 0\\ 0& 0& -1 \end{smallmatrix}\right]\\
\hline
\mathbf{(1432)}& 1& -1& \left[\begin{smallmatrix}
0& \omega\\ \omega ^2& 0
\end{smallmatrix}\right]& \tiny\left[\begin{smallmatrix}
0& 1& 0\\ -1& 0& 0\\ 0& 0& 1 \end{smallmatrix}\right]&
\tiny\left[\begin{smallmatrix}
0& -1& 0\\ 1& 0& 0\\ 0& 0& -1 \end{smallmatrix}\right]\\
\hline
\mathbf{(1423)}& 1& -1& \left[\begin{smallmatrix}
0& 1\\ 1& 0
\end{smallmatrix}\right]& \tiny\left[\begin{smallmatrix}
1& 0& 0\\ 0& 0& -1\\ 0& 1& 0 \end{smallmatrix}\right]&
\tiny\left[\begin{smallmatrix}
-1& 0& 0\\ 0& 0& 1\\ 0& -1& 0 \end{smallmatrix}\right]\\
\hline
\mathbf{(1342)}& 1& -1& \left[\begin{smallmatrix}
0& \omega ^2\\ \omega& 0
\end{smallmatrix}\right]& \tiny\left[\begin{smallmatrix}
0& 0& 1\\ 0& 1& 0\\ -1& 0& 0 \end{smallmatrix}\right]&
\tiny\left[\begin{smallmatrix}
0& 0& -1\\ 0& -1& 0\\ 1& 0& 0 \end{smallmatrix}\right]\\
\hline
\mathbf{(1324)}& 1& -1& \left[\begin{smallmatrix}
0& 1\\ 1& 0
\end{smallmatrix}\right]& \tiny\left[\begin{smallmatrix}
1& 0& 0\\ 0& 0& 1\\ 0& -1& 0 \end{smallmatrix}\right]&
\tiny\left[\begin{smallmatrix}
-1& 0& 0\\ 0& 0& -1\\ 0& 1& 0 \end{smallmatrix}\right]\\
\hline
\mathbf{(1243)}& 1& -1& \left[\begin{smallmatrix}
0& \omega ^2\\ \omega& 0
\end{smallmatrix}\right]& \tiny\left[\begin{smallmatrix}
0& 0& -1\\ 0& 1& 0\\ 1& 0& 0 \end{smallmatrix}\right]&
\tiny\left[\begin{smallmatrix}
0& 0& 1\\ 0& -1& 0\\ -1& 0& 0 \end{smallmatrix}\right]\\
\hline
\end{array}$
}
\caption{The unitary representations of $\bbS_4$ in matrix form}
\label{table:unitary-rep}
\end{table}
\FloatBarrier

%
\begin{table}
\resizebox{15cm}{5cm}
{
$\begin{array}{|c|c|c|c|c|c|c|c|c|c|c|c|c|c|c|c|c|c|c|c|c|c|c|c|c|}
\hline
\mathbf{\bbS_4}& \boldsymbol\iota_{1,1}& \boldsymbol\tau_{1,1}& \boldsymbol\sigma_{1,1}& \boldsymbol\sigma_{2,1}& \boldsymbol\sigma_{1,2}& \boldsymbol\sigma_{2,2}& \boldsymbol\pi_{1,1}& \boldsymbol\pi_{2,1}& \boldsymbol\pi_{3,1}& \boldsymbol\pi_{1,2}& \boldsymbol\pi_{2,2}& \boldsymbol\pi_{3,2}& \boldsymbol\pi_{1,3}& \boldsymbol\pi_{2,3}& \boldsymbol\pi_{3,3}& \boldsymbol\pi'_{1,1}& \boldsymbol\pi'_{2,1}& \boldsymbol\pi'_{3,1}& \boldsymbol\pi'_{1,2}& \boldsymbol\pi'_{2,2}& \boldsymbol\pi'_{3,2}& \boldsymbol\pi'_{1,3}& \boldsymbol\pi'_{2,3}& \boldsymbol\pi'_{3,3}\\
\hline
\textbf{id}& 1& 1& 1& 0& 0& 1& 1& 0& 0& 0& 1& 0& 0& 0& 1& 1& 0& 0& 0& 1& 0& 0& 0& 1\\
\hline
\mathbf{(12)}& 1& -1& 0& 1& 1& 0& -1& 0& 0& 0& 0& 1& 0& 1& 0& 1& 0& 0& 0& 0& -1& 0& -1& 0\\
\hline
\mathbf{(23)}& 1& -1& 0& \omega ^2& \omega& 0& 0& 0& -1& 0& -1& 0& -1& 0& 0& 0& 0& 1& 0& 1& 0& 1& 0& 0\\
\hline
\mathbf{(34)}& 1& -1& 0& 1& 1& 0& -1& 0& 0& 0& 0& -1& 0& -1& 0& 1& 0& 0& 0& 0& 1& 0& 1& 0\\
\hline
\mathbf{(13)}& 1& -1& 0& \omega& \omega ^2& 0& 0& 1& 0& 1& 0& 0& 0& 0& -1& 0& -1& 0& -1& 0& 0& 0& 0& 1\\
\hline
\mathbf{(14)}& 1& -1& 0& \omega ^2& \omega& 0& 0& 0& 1& 0& -1& 0& 1& 0& 0& 0& 0& -1& 0& 1& 0& -1& 0& 0\\
\hline
\mathbf{(24)}& 1& -1& 0& \omega& \omega ^2& 0& 0& -1& 0& -1& 0& 0& 0& 0& -1& 0& 1& 0& 1& 0& 0& 0& 0& 1\\
\hline
\mathbf{(12)(34)}& 1& 1& 1& 0& 0& 1& 1& 0& 0& 0& -1& 0& 0& 0& -1& 1& 0& 0& 0& -1& 0& 0& 0& -1\\
\hline
\mathbf{(13)(24)}& 1& 1& 1& 0& 0& 1& -1& 0& 0& 0& -1& 0& 0& 0& 1& -1& 0& 0& 0& -1& 0& 0& 0& 1\\
\hline
\mathbf{(14)(23)}& 1& 1& 1& 0& 0& 1& -1& 0& 0& 0& 1& 0& 0& 0& -1& -1& 0& 0& 0& 1& 0& 0& 0& -1\\
\hline
\mathbf{(123)}& 1& 1& \omega& 0& 0& \omega ^2& 0& 0& 1& -1& 0& 0& 0& -1& 0& 0& 0& 1& -1& 0& 0& 0& -1& 0\\
\hline
\mathbf{(132)}& 1& 1& \omega ^2& 0& 0& \omega& 0& -1& 0& 0& 0& -1& 1& 0& 0& 0& -1& 0& 0& 0& -1& 1& 0& 0\\
\hline
\mathbf{(124)}& 1& 1& \omega ^2& 0& 0& \omega& 0& 1& 0& 0& 0& -1& -1& 0& 0& 0& 1& 0& 0& 0& -1& -1& 0& 0\\
\hline
\mathbf{(142)}& 1& 1& \omega& 0& 0& \omega ^2& 0& 0& -1& 1& 0& 0& 0& -1& 0& 0& 0& -1& 1& 0& 0& 0& -1& 0\\
\hline
\mathbf{(134)}& 1& 1& \omega& 0& 0& \omega ^2& 0& 0& -1& -1& 0& 0& 0& 1& 0& 0& 0& -1& -1& 0& 0& 0& 1& 0\\
\hline
\mathbf{(143)}& 1& 1& \omega ^2& 0& 0& \omega& 0& -1& 0& 0& 0& 1& -1& 0& 0& 0& -1& 0& 0& 0& 1& -1& 0& 0\\
\hline
\mathbf{(234)}& 1& 1& \omega^2& 0& 0& \omega& 0& 1& 0& 0& 0& 1& 1& 0& 0& 0& 1& 0& 0& 0& 1& 1& 0& 0\\
\hline
\mathbf{(243)}& 1& 1& \omega& 0& 0& \omega ^2& 0& 0& 1& 1& 0& 0& 0& 1& 0& 0& 0& 1& 1& 0& 0& 0& 1& 0\\
\hline
\mathbf{(1234)}& 1& -1& 0& \omega& \omega ^2& 0& 0& -1& 0& 1& 0& 0& 0& 0& 1& 0& 1& 0& -1& 0& 0& 0& 0& -1\\
\hline
\mathbf{(1432)}& 1& -1& 0& \omega& \omega ^2& 0& 0& 1& 0& -1& 0& 0& 0& 0& 1& 0& -1& 0& 1& 0& 0& 0& 0& -1\\
\hline
\mathbf{(1423)}& 1& -1& 0& 1& 1& 0& 1& 0& 0& 0& 0& -1& 0& 1& 0& -1& 0& 0& 0& 0& 1& 0& -1& 0\\
\hline
\mathbf{(1342)}& 1& -1& 0& \omega ^2& \omega& 0& 0& 0& 1& 0& 1& 0& -1& 0& 0& 0& 0& -1& 0& -1& 0& 1& 0& 0\\
\hline
\mathbf{(1324)}& 1& -1& 0& 1& 1& 0& 1& 0& 0& 0& 0& 1& 0& -1& 0& -1& 0& 0& 0& 0& -1& 0& 1& 0\\
\hline
\mathbf{(1243)}& 1& -1& 0& \omega ^2& \omega& 0& 0& 0& -1& 0& 1& 0& 1& 0& 0& 0& 0& 1& 0& -1& 0& -1& 0& 0\\
\hline
\end{array}$
}
\caption{The matrix coefficients of the unitary representations of $\bbS_4$}
\label{table:unitary-rep-matrix-coef}
\end{table}
\FloatBarrier

\newcommand{\etalchar}[1]{$^{#1}$}

\end{document}